\documentclass[runningheads]{llncs}
\pdfoutput=1

\usepackage{graphicx}
\usepackage{wrapfig}

\usepackage{subcaption}
\usepackage{makecell}

\usepackage{amsmath}
\usepackage{amssymb}

\usepackage{mathtools}
\usepackage{algorithm}
\usepackage{algpseudocode}

\usepackage{stackengine}

\usepackage{enumitem}

\usepackage{tikz}

\usepackage{hyperref} 
\usepackage{cleveref}

\usepackage{ stmaryrd }
\usepackage{enumitem}
\usepackage{wrapfig}
\usepackage{varwidth}

\usepackage{bussproofs}
\usepackage{booktabs}

\usepackage{marvosym}

\usepackage{thm-restate} 
\usepackage[skins,theorems]{tcolorbox}
\tcbset{highlight math style={enhanced,
colframe=red,colback=white,arc=0pt,boxrule=1pt}}
\usepackage{ltl}

\usepackage{array}
\usepackage{arydshln}

\usepackage{orcidlink}

\usepackage[firstpage]{draftwatermark}  

\SetWatermarkAngle{0}
\SetWatermarkText{\raisebox{12.5cm}{%
		\hspace{0.1cm}%
		\href{https://doi.org/10.6084/m9.figshare.19697656}{\includegraphics{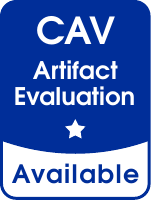}}%
		\hspace{9cm}%
		\includegraphics{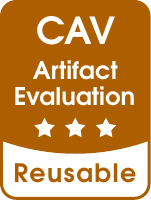}%
}}

\newenvironment{proofSketch}{%
	\proof}{\endproof}

\newcommand{\ltlN}{\LTLnext}
\newcommand{\ltlG}{\LTLglobally}
\newcommand{\ltlF}{\LTLeventually}
\newcommand{\ltlU}{\LTLuntil}

\newcommand{\calT}{\mathcal{T}}
\newcommand{\calA}{\mathcal{A}}

\newcommand{\traces}[1]{\mathit{Traces}(#1)}
\newcommand{\traceVars}{\mathcal{V}}
\newcommand{\nat}{\mathbb{N}}

\newcommand{\filter}[2]{\llparenthesis #1\rrparenthesis_{#2}}
\newcommand{\predSat}[1]{\llbracket #1 \rrbracket}

\newcommand{\validRes}[3]{\mathit{validRes}^{#1,#2}_{#3}}

\newcommand{\gameE}[3]{\mathcal{G}^\forall_{(#1, #2, #3)}}
\newcommand{\gameEF}[3]{\mathcal{G}^{\forall\exists}_{(#1, #2, #3)}}
\newcommand{\STS}{\mathcal{T}}
\newcommand{\pred}{\mathcal{P}}

\algsetblockdefx[Loopi]{WhileOne}{}
{1}{\algorithmicindent/2}
[1]{\textbf{while} #1 \textbf{do}}

\algsetblockdefx[Loopi]{WhileTwo}{}
{2}{\algorithmicindent/2}
[1]{\textbf{while} #1 \textbf{do}}

\algsetblockdefx[Loopi]{IfOne}{}
{1}{\algorithmicindent/2}
[1]{\textbf{if} #1 \textbf{then}}

\algsetblockdefx[Loopi]{ElseOne}{}
{1}{\algorithmicindent/2}
[0]{\textbf{else}}

\algsetblockdefx[Loopi]{IfTwo}{}
{2}{\algorithmicindent/2}
[1]{\textbf{if} #1 \textbf{then}}

\algsetblockdefx[Loopi]{IfThree}{}
{3}{\algorithmicindent}
[1]{\textbf{if} #1 \textbf{then}}

\algsetblockdefx[Loopi]{ElseThree}{}
{3}{\algorithmicindent/2}
[0]{\textbf{else}}

\algsetblockdefx[Loopi]{ElseTwo}{}
{2}{\algorithmicindent/2}
[0]{\textbf{else}}

\algsetblockdefx[Loopi]{RepeatOne}{}
{1}{\algorithmicindent/2}
[0]{\textbf{repeat}}

\algsetblockdefx[Loopi]{RepeatTwo}{}
{2}{\algorithmicindent/2}
[0]{\textbf{repeat}}

\algsetblockdefx[Loopi]{RepeatThree}{}
{3}{\algorithmicindent/2}
[0]{\textbf{repeat}}

\algsetblockdefx[Loopi]{RepeatFour}{}
{4}{\algorithmicindent/2}
[0]{\textbf{repeat}}

\algsetblockdefx[Loopi]{RepeatFive}{}
{5}{\algorithmicindent/2}
[0]{\textbf{repeat}}

\algsetblockdefx[Loopi]{ForAllOne}{}
{1}{\algorithmicindent/2}
[1]{\textbf{for all} #1 \textbf{do}}

\newcommand{\ldot}{\mathpunct{.}}

\algsetblockdefx[Loopi]{CaseTwo}{}
{2}{\algorithmicindent/2}
[1]{\textbf{case} #1\textbf{:}}

\algsetblockdefx[Loopi]{CaseThree}{}
{3}{\algorithmicindent/2}
[1]{\textbf{case} #1\textbf{:}}

\algsetblockdefx[Loopi]{CaseFour}{}
{4}{\algorithmicindent/2}
[1]{\textbf{case} #1\textbf{:}}

\algsetblockdefx[Loopi]{MatchTwo}{}
{2}{\algorithmicindent/2}
[1]{\textbf{match } #1 \textbf{with}}

\makeatletter
\newcommand\fs@plainruled{\def\@fs@cfont{\rmfamily}\let\@fs@capt\floatc@plain
	\def\@fs@pre{\hrule height1pt depth0pt \kern5pt}%
	\def\@fs@post{}%
	\def\@fs@mid{\kern2pt\hrule height1pt depth0pt\relax\kern\abovecaptionskip}%
	\let\@fs@iftopcapt\iffalse}
\makeatother

\newif\iffullversion
\fullversiontrue

\newcommand{\ifFull}[2]{\iffullversion#1\else#2\fi}

\begin{document}

\title{Software Verification of Hyperproperties \\Beyond $k$-Safety}
\titlerunning{Software Verification of Hyperproperties Beyond $k$-Safety}

\author{Raven Beutner\textsuperscript{(\Letter)} \!\!\! \scalebox{1.25}{\orcidlink{0000-0001-6234-5651}}  \and
Bernd Finkbeiner \!\!\! \scalebox{1.25}{\orcidlink{0000-0002-4280-8441}} }
\authorrunning{R.~Beutner and B.~Finkbeiner}

\institute{CISPA Helmholtz Center for Information Security,\\ Saarbrücken, Germany\\ \email{\{raven.beutner, finkbeiner\}@cispa.de}}

\maketitle              

\begin{abstract}
Temporal hyperproperties are system properties that relate multiple execution traces. 
For (finite-state) hardware, temporal hyperproperties are supported by model checking algorithms, and tools for general temporal logics like HyperLTL exist. 
For (infinite-state) software, the analysis of temporal hyperproperties has, so far, been limited to $k$-safety properties, i.e., properties that stipulate the absence of a bad interaction between any $k$ traces.
In this paper, we present an automated method for the verification of $\forall^k\exists^l$-safety properties in infinite-state systems.
A $\forall^k\exists^l$-safety property stipulates that for any $k$ traces, there \emph{exist} $l$ traces such that the resulting $k+l$ traces do not interact badly.
This combination of universal and existential quantification enables us to express many properties beyond $k$-safety, including, for example, generalized non-interference or program refinement.
Our method is based on a strategy-based instantiation of existential trace quantification combined with a program reduction, both in the context of a fixed predicate abstraction.
Notably, our framework allows for mutual dependence of strategy and reduction.

\end{abstract}

\section{Introduction}

Hyperproperties are system properties that relate multiple execution traces of a system \cite{ClarksonS08} and commonly arise, e.g., in information-flow policies \cite{McCullough88}, the verification of code optimizations \cite{BarrettFGHPZ05}, and robustness of software \cite{ChaudhuriGL12}.
Consequently, many methods for the automated verification of hyperproperties have been developed \cite{ShemerGSV19,UnnoTK21,FarzanV19,SousaD16}.
Almost all previous approaches verify a class of hyperproperties called $k$-safety, i.e., properties that stipulate the absence of a bad interaction between any $k$ traces in the system.
For example, we can express a simple form of non-interference as a $2$-safety property by stating that any \emph{two} traces that agree on the low-security inputs should produce the same observable output.

The vast landscape of hyperproperties does, however, stretch far beyond $k$-safety.
The overarching limitation of $k$-safety (or, more generally, of hypersafety \cite{ClarksonS08}) is an implicit \emph{universal} quantification over all executions.
By contrast, many properties of interest, ranging from applications in information-flow control to robust cleanness, require a combination of universal and existential quantification.
For example, consider the reactive program in \Cref{fig:gniex}, where $\star_{\nat}$ denotes a nondeterministic choice of a natural number.
We assume that $h$, $l$, and $o$ are a high-security input, a low-security input, and a low-security output, respectively.
This program violates the simple $2$-safety non-interference property given above as the non-determinism influences the output. 
Nevertheless, the program is ``secure'' in the sense that an attacker that observes low-security inputs and outputs cannot deduce information about the high-security input.
To capture this formally, we use a relaxed notion of non-interference, in the literature often referred to as generalized non-interference (GNI) \cite{McCullough88}.
We can, informally, express GNI in a temporal logic as follows:
\begin{align*}
	\forall \pi. \forall \pi'. \exists \pi''\ldot \ltlG \big( o_\pi = o_{\pi''} \land l_\pi = l_{\pi''}  \land h_{\pi'} = h_{\pi''} \big)
\end{align*}
This property requires that for any two traces $\pi, \pi'$, there exists some trace $\pi''$ that, globally, agrees with the low-security inputs and outputs on $\pi$ but the high-security inputs on $\pi'$.
Phrased differently, any observation on the low-security input-output behavior is compatible with every possible high-security input. 
The program in \Cref{fig:gniex} satisfies GNI.
Crucially, GNI  is no longer a hypersafety property (and, in particular, no $k$-safety property for any $k$) as it requires a combination of universal and \emph{existential} quantification.

\setlength{\columnsep}{0pt}%
\setlength{\intextsep}{8pt}%
\begin{wrapfigure}{r}{0.25\textwidth}
    \vspace{-5mm}
    \small
    \begin{algorithmic}
        \RepeatFive
        \State \textbf{readInput}($h, l$)
        \IfOne{$h > l$}
        \State $o \leftarrow l + \star_{\nat}$
        \ElseThree
        \State $x \leftarrow \star_{\nat}$
        \IfOne{$x \geq l$}
        \State $o \leftarrow x$
        
        \ElseOne
        \State $o \leftarrow l$
    \end{algorithmic}
	\vspace{-1mm}
    \caption{}
    \label{fig:gniex}
\end{wrapfigure}

\subsection{Verification Beyond $k$-Safety}

Instead, GNI falls in the general class of $\forall^*\exists^*$-safety properties.
Concretely, a $\forall^k\exists^l$-safety property (using $k$ universal and $l$ existential quantifiers) stipulates that for any $k$ traces, there exist $l$ traces such that the resulting $k+l$ traces do not interact badly. 
$k$-safety properties are the \emph{special case} where $l = 0$.
We study the verification of such properties in infinite-state systems arising, e.g., in software.
In contrast to $k$-safety, where a broad range of methods has been developed \cite{ShemerGSV19,FarzanV19,SousaD16,Benton04,UnnoTK21}, no method for the automated verification of \emph{temporal} $\forall^*\exists^*$ properties in infinite-state systems exists (we discuss related approaches in \Cref{sec:relatedWork}).

Our novel verification method is based on a game-based reading of existential quantification \emph{combined} with the search for a program reduction.
The game-based reading of existential quantification instantiates existential trace quantification with an explicit strategy and constitutes the first practicable method for the verification of $\forall^*\exists^*$-properties in finite-state systems \cite{CoenenFST19}.
Program reductions are a well-established technique to align executions of independent program fragments (such as the individual program copies in a self-composition) to obtain proofs with easier invariants \cite{Lipton75,ShemerGSV19,FarzanV19}.

So far, both techniques are limited to their respective domain, i.e., the game-based approach has only been applied to finite-state systems and synchronous specifications, and reductions have (mostly) been used for the verification of $k$-safety.
We combine both techniques yielding an effective (and first) verification technique for hyperproperties beyond $k$-safety in infinite-state systems arising in software.
Notably, our search for reduction and strategy-based instantiation of existential quantification is \emph{mutually dependent}, i.e., a particular strategy might depend on a particular reduction and vice versa.

\subsection{Contributions and Structure}

The starting point of our work is a new temporal logic called \emph{Observation-based HyperLTL} (OHyperLTL for short).
Our logic extends the existing hyperlogic HyperLTL \cite{ClarksonFKMRS14} with capabilities to reason about asynchronous properties (i.e., properties where the individual traces are traversed at different speeds), and to specify properties using assertions from arbitrary background theories (to reason about the infinite domains encountered in software) (\Cref{sec:ohyperltl}).

To automatically verify $\forall^k\exists^l$ OHyperLTL properties, we combine program reductions with a strategy-based instantiation of existential quantification, both in the context of a fixed predicate abstraction.
To facilitate this combination, we first present a game-based approach that automates the search for a reduction.
Concretely, we construct an abstract game where a winning strategy for the verifier directly corresponds to a reduction with accompanying proof.
As a side product, our game-based interpretation simplifies the search for a reduction in a given predicate abstraction as, e.g., studied by Shemer et al.~\cite{ShemerGSV19} (\Cref{sec:ksafety}).

Our strategic (game-based) view on reductions allows us to combine them with a game-based instantiation of existential quantification. 
Here, we view the existentially quantified traces as being constructed by a strategy that, iteratively, reacts to the universally quantified traces. 
As we phrase both the search for a reduction and the search for existentially quantified traces as a game, we can frame the search for both as a combined abstract game.
We prove the soundness of our approach, i.e., a winning strategy for the verifier constitutes both a strategy for the existentially quantified traces and accompanying (mutually dependent) reduction.
Despite its finite nature, constructing the abstract game is expensive as it involves many SMT queries.
We propose an inner refinement loop that determines the winner of the game (without constructing it explicitly) by computing iterative approximations (\Cref{sec:beyondksafety}).

We have implemented our verification approach in a prototype tool called \texttt{HyPA} (short for \textbf{Hy}perproperty Verification with \textbf{P}redicate \textbf{A}bstraction)
and evaluate \texttt{HyPA} on $k$-safety properties (that can already be handled by existing methods) and on $\forall^*\exists^*$-safety benchmarks that cannot be handled by any existing tool (\Cref{sec:eval}).

\paragraph{Contributions.}
In short, our contributions include the following:

\begin{itemize}[leftmargin=*]
	\item We propose a temporal hyperlogic that can specify asynchronous hyperproperties in infinite-state systems;
	
	\item We propose a game-based interpretation of a reduction (improving and simplifying previous methods for $k$-safety \cite{ShemerGSV19});
	
	\item We combine a strategy-based instantiation of existentially quantified traces with the search for a reduction. This yields a flexible (and first) method for the verification of temporal $\forall^*\exists^*$ properties.
	We propose an iterative method to solve the abstract game that avoids an expensive explicit construction;
	
	\item We provide and evaluate a prototype implementation of our method.
\end{itemize}

\section{Overview: Reductions and Quantification as a Game}
\label{sec:overview}

Our verification approach hinges on the observation that we can express both a reduction and existential trace quantification as a game.
In this section, we provide an overview of our game-based interpretations.
We begin by outlining our game-based reading of a reduction (illustrating this in the simpler case of $k$-safety) in \Cref{sec:overview_ksafety} and then extend this to include a game-based interpretation of existential quantification in \Cref{sec:overview_beyond}.

\subsection{Reductions as a Game}
\label{sec:overview_ksafety}

Consider the two programs in \Cref{fig:univ} and the specification that both programs produce the same output (on initially identical values for $x$).
We can formalize this in our logic OHyperLTL (formally defined in \Cref{sec:ohyperltl}) as follows: 
\begin{align*}
\forall^{\texttt{P1}} \pi_1 :(\mathit{pc} = 2). \; \forall^{\texttt{P2}} \pi_2 : (\mathit{pc} = 2). \; (x_{\pi_1} = x_{\pi_2}) \rightarrow \ltlG  (x_{\pi_1} = x_{\pi_2})
\end{align*}
The property states that for all traces $\pi_1$ in \texttt{P1} and $\pi_2$ in \texttt{P2} the LTL specification $(x_{\pi_1} = x_{\pi_2}) \rightarrow \ltlG  (x_{\pi_1} = x_{\pi_2})$ holds (where $x_\pi$ refers to the value of $x$ on trace $\pi$).
Additionally, the observation formula $\mathit{pc} = 2$ marks the positions at which the LTL property is evaluated:
We only observe a trace at steps where $\mathit{pc} = 2$ (i.e., where the program counter is at the output position). 

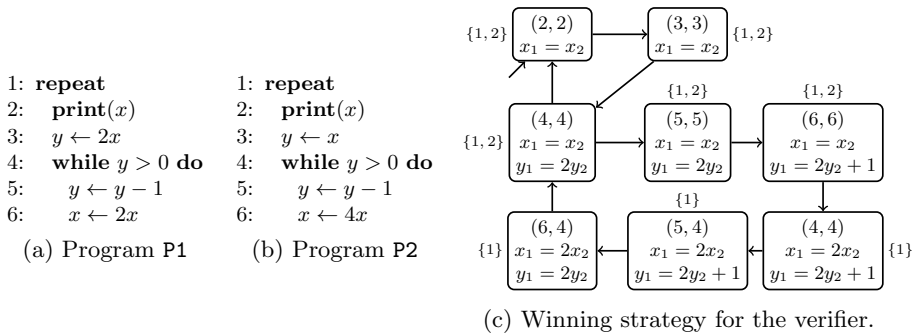
\begin{figure}[t]
        \begin{subfigure}{0.25\textwidth}
            \scalebox{0.9}{
                \begin{minipage}{\textwidth}
                    \begin{algorithmic}[1]
                        \RepeatThree
                        \State \textbf{print$(x)$}
                        \State $y \leftarrow 2  x$
                        \WhileTwo{$y > 0$}
                        \State $y \leftarrow y - 1$
                        \State $x \leftarrow 2x$
                    \end{algorithmic}
                \end{minipage}
            }
            \subcaption{Program \texttt{P1}}\label{fig:p1}
        \end{subfigure}%
        \begin{subfigure}{0.25\textwidth}
            \scalebox{0.9}{
                \begin{minipage}{\textwidth}
                    \begin{algorithmic}[1]
                        \RepeatThree
                        \State \textbf{print$(x)$}
                        \State $y \leftarrow x$
                        \WhileTwo{$y > 0$}
                        \State $y \leftarrow y - 1$
                        \State $x \leftarrow 4x$
                        
                    \end{algorithmic}
                \end{minipage}
            }
            \subcaption{Program \texttt{P2}}\label{fig:p2}
        \end{subfigure}%
        \begin{subfigure}{0.5\textwidth}
            \centering
            \scalebox{0.8}{
                \small
                \begin{tikzpicture}[scale=0.9]
                    
                    \node[draw, rectangle, thick,align=center, rounded corners=3pt,label={[label distance=-0.5mm]west:{\scriptsize$\{1,2\}$}}] at (2.5,2) (n0) {$(2,2)$\\ $x_1 = x_2$};
                    
                    \node[draw, rectangle, thick,align=center, rounded corners=3pt,label={[label distance=-0.5mm]west:{\scriptsize$\{1,2\}$}}] at (2.5,0) (n1) {$(4,4)$\\ $x_1 = x_2$\\$y_1 = 2 y_2$};
                    
                    \node[draw, rectangle, thick,align=center, rounded corners=3pt,label={[label distance=-0.5mm]north:{\scriptsize$\{1,2\}$}}] at (5,0) (n2) {$(5,5)$\\ $x_1 = x_2$\\$y_1 = 2 y_2$};
                    
                    \node[draw, rectangle, thick,align=center, rounded corners=3pt,label={[label distance=-0.5mm]north:{\scriptsize$\{1,2\}$}}] at (7.5,0) (n3) {$(6,6)$\\ $x_1 = x_2$\\$y_1 = 2 y_2 + 1$};
                    
                    \node[draw, rectangle, thick,align=center, rounded corners=3pt,label={[label distance=-0.5mm]east:{\scriptsize$\{1,2\}$}}] at (5,2) (n4) {$(3,3)$\\ $x_1 = x_2$};
                    
                    \node[draw, rectangle, thick,align=center, rounded corners=3pt,label={[label distance=-0.5mm]west:{\scriptsize$\{1\}$}}] at (2.5,0-2) (n5) {$(6,4)$\\ $x_1 = 2x_2$\\$y_1 = 2 y_2$};
                    
                    \node[draw, rectangle, thick,align=center, rounded corners=3pt,label={[label distance=-0.5mm]north:{\scriptsize$\{1\}$}}] at (5,-2) (n6) {$(5,4)$\\ $x_1 = 2x_2$\\$y_1 = 2 y_2 + 1$};
                    
                    \node[draw, rectangle, thick,align=center, rounded corners=3pt,label={[label distance=-0.5mm]east:{\scriptsize$\{1\}$}}] at (7.5,-2) (n7) {$(4,4)$\\ $x_1 = 2x_2$\\$y_1 = 2 y_2 + 1$};

                    \draw[->, thick] (n1) --  (n2);
                    
                    \draw[->, thick] (n2) -- (n3);
                    
                    \draw[->, thick] (n3) -- (n7);
                    
                    \draw[->, thick] (n7) -- (n6);
                    
                    \draw[->, thick] (n6) -- (n5);
                    
                    \draw[->, thick] (n5) --(n1);

                    \draw[->, thick] (n0) --  (n4);
                    
                    \draw[->, thick] (n4) -- (n1);
                    
                    \draw[->, thick] (n1) -- (n0);

                    \draw[->, thick] (n0) + (-0.8, -0.8) -- (n0);
                \end{tikzpicture}
            }
            \subcaption{Winning strategy for the verifier.} \label{fig:winningStrat}
        \end{subfigure}

    \caption{Two output-equivalent programs \texttt{P1} and \texttt{P2} are depicted in \Cref{fig:p1,fig:p2}. In \Cref{fig:winningStrat} a possible winning strategy for the verifier is given.
        Each abstract state contains the value of the program counter of both copies (given as the pair at the top) and the predicates that hold in that state.  
        For sake of readability we omit the trace variables and write, e.g., $x_1$ for $x_{\pi_1}$.
        We mark the initial state with an incoming arrow.
        The outer label at each state gives the scheduling $M \subseteq \{1, 2\}$ chosen by the strategy in that state.\vspace{-0.3cm}}\label{fig:univ}
\end{figure}

The verification of our property involves reasoning about two copies of our system (in this case, one of \texttt{P1} and one of \texttt{P2}) on \emph{disjoint} state spaces.
Consequently, we can interleave the statements of both programs (between two observation points) without affecting the behavior of the individual copies.
We refer to each interleaving of both copies as a \emph{reduction}.
The choice of a reduction drastically influences the complexity of the needed invariants \cite{FarzanV19,ShemerGSV19,Lipton75}.
Given an initial abstraction of the system \cite{GrafS97,ShemerGSV19}, we aim to discover a suitable reduction \emph{automatically}.
Our first observation is that we can phrase the search for a reduction as a game as follows:
In each step, the verifier decides on a \emph{scheduling} (i.e., a non-empty subset $M \subseteq \{1, 2\}$) that indicates which of the copies should take a step (i.e., $i \in M$ iff copy $i$ should make a program step).
Afterward, the refuter can choose an abstract successor state compatible with that scheduling, after which the process repeats.
This naturally defines a finite-state two-player safety game that we can solve efficiently.\footnote{The LTL specification is translated to a symbolic safety automaton that moves alongside the game. For sake of readability, we omitted the automaton from the following discussion.}
If the verifier wins, a winning strategy directly corresponds to a reduction and accompanying inductive invariant for the safety property within the given abstraction.

For our example, we give (parts of) a possible winning strategy in \Cref{fig:winningStrat}.
In each abstract state, the strategy chooses a scheduling (written next to the state), and all abstract states compatible with that scheduling are listed as successors.
Note that whenever the program counter is $(2,2)$ (i.e., both programs are at their output position), it holds that $x_1 = x_2$ (as required).
The example strategy schedules in lock-step for the most part (by choosing $M = \{1, 2\}$) but lets \texttt{P1} take the inner loop \emph{twice}, thereby maintaining the linear invariants $x_1 = x_2$ and $y_1 = 2y_2$. 
In particular, the resulting reduction is property-based \cite{ShemerGSV19} as the scheduling is based on the current (abstract) state. 
Note that the program cannot be verified with only linear invariants in a sequential or parallel (lock-step) reduction.

\subsection{Beyond $k$-Safety: Quantification as a Game}
\label{sec:overview_beyond}

We build upon this game-based interpretation of a reduction to move beyond $k$-safety.
As a second example, consider the two programs \texttt{Q1} and \texttt{Q2} in \Cref{fig:exist}, where $\star_\tau$ denotes a nondeterministic choice of type $\tau \in \{\nat, \mathbb{B}\}$.
We wish to check that \texttt{Q1} refines \texttt{Q2}, i.e., all output behavior of \texttt{Q1} is also possible in \texttt{Q2}.
We can express this in our logic as follows:
\begin{align*}
\forall^{\texttt{Q1}} \pi_1 : (\mathit{pc} = 2). \; \exists^{\texttt{Q2}} \pi_2 : (\mathit{pc} = 2). \; \ltlG (a_{\pi_1} = a_{\pi_2})
\end{align*}%
The property states that for every trace $\pi_1$ in \texttt{Q1} there \emph{exists} a trace $\pi_2$ in \texttt{Q2} that outputs the same value. 
The quantifiers range over infinite traces of variable assignments (with infinite domains), making a direct verification of the quantifier alternation challenging.
In contrast to alternation-free formulas, we cannot reduce the verification to verification on a self composition \cite{BartheDR11,FinkbeinerRS15}.
Instead, we adopt (yet another) game-based interpretation by viewing the existentially quantified traces as being resolved by a \emph{strategy} (called the witness strategy) \cite{CoenenFST19}.
That is, instead of trying to find a witness traces $\pi_2$ in \texttt{Q2} when given the \emph{entire} trace $\pi_1$, we interpret the $\forall\exists$ property as a game between verifier and refuter. 
The refuter moves through the state space of \texttt{Q1} (thereby producing a trace $\pi_1$), and the verifier reacts to each move by choosing a successor in the state space of \texttt{Q2} (thereby producing a trace $\pi_2$).
If the verifier can assure that the resulting  traces $\pi_1, \pi_2$ satisfy $ \ltlG (a_{\pi_1} = a_{\pi_2})$, the $\forall\exists$ property holds. 
However, this game-based interpretation fails in many instances. 
There might exist a witness trace $\pi_2$, but the trace cannot be produced by a witness strategy as it requires knowledge of \emph{future} moves of the refuter. 
Let us discuss this on the example programs in \Cref{fig:exist}.
A simple (informal) solution to construct a witness trace $\pi_2$ (when given the entire $\pi_1$) would be to guarantee that in \texttt{Q2}:4 (meaning location 4 of \texttt{Q2}) and line \texttt{Q1}:6 the value of $x$ in both programs agrees (i.e., $x_1 = x_2$ holds) and then simply resolve the nondeterminism at \texttt{Q2}:6 with $0$.
However, to follow this idea, the witness strategy for the verifier, when at \texttt{Q2}:3, would need to know the future value of $x_1$ when \texttt{Q1} is at location \texttt{Q1}:6. 

\begin{figure}[t]   
    \begin{subfigure}{0.25\textwidth}
        \scalebox{0.9}{
            \begin{minipage}{\textwidth}
                \begin{algorithmic}[1]
                    \RepeatFive
                    \State \textbf{print($a$)}
                    \State $x \leftarrow \star_\nat$
                    \WhileOne{$\star_\mathbb{B}$}
                    \State $x \leftarrow x + 1$
                    
                    \State $y \leftarrow x$
                    \WhileTwo{$y > 0$}
                    \State $a \leftarrow a + x$
                    \State $y \leftarrow y - 1$
                \end{algorithmic}
            \end{minipage}
        }
        \subcaption{Program \texttt{Q1}}\label{fig:p1'}
    \end{subfigure}%
    \begin{subfigure}{0.25\textwidth}
        \scalebox{0.9}{
            \begin{minipage}{\textwidth}
                \begin{algorithmic}[1]
                    \RepeatFour
                    \State \textbf{print($a$)}
                    \State $x \leftarrow \star_\nat$
                    \State $y \leftarrow x$
                    \WhileTwo{$y > 0$}
                    \State $a \leftarrow a +  x + \star_\nat$
                    \State $y \leftarrow y - 1$
                \end{algorithmic}
            \end{minipage}
        }
        \subcaption{Program \texttt{Q2}}\label{fig:p2'}
    \end{subfigure}%
    \begin{subfigure}{0.5\textwidth}
        \scalebox{0.85}{
            \small
            \centering
            \begin{tikzpicture}[scale=0.9]
                \node[draw, rectangle, thick,align=center, rounded corners=3pt, fill=white,label=-45:{\scriptsize$ \{1\}$}] at (0,0) (n0) {$\boldsymbol{\alpha_1:}(5,3)$\\ $a_1 = a_2$};
                
                \node[draw, rectangle, thick,align=center, rounded corners=3pt, fill=white,label=south:{\scriptsize$\{1\}$}] at (2.5,0) (n1) {$\boldsymbol{\alpha_2:}(3,3)$\\ $a_1 = a_2$};
                
                \node[draw, rectangle, thick,align=center, rounded corners=3pt, fill=white,label=south:{\scriptsize$\{1, 2\}$}] at (5,0) (n2) {$\boldsymbol{\alpha_3:}(2,2)$\\ $a_1 = a_2$};
                
                \node[draw, rectangle, thick,align=center, rounded corners=3pt, fill=white,label={[xshift=-0.4cm]south:{\scriptsize$\{1\}$}}] at (0,-2.5) (n3) {$\boldsymbol{\alpha_4:}(4,3)$\\ $a_1 = a_2$};

                \node[draw, rectangle, thick,align=center, rounded corners=3pt, fill=white,label=north:{\scriptsize$\{1, 2\}$}] at (2.5,-2.5) (n4) {$\boldsymbol{\alpha_5:}(7,5)$\\ $a_1 = a_2$\\$x_1 = x_2$\\$y_1 = y_2$};
                
                \node[draw, rectangle, thick,align=center, rounded corners=3pt, fill=white,label={[align=center]north:{\scriptsize$\{1,2\}$}}] at (5,-2.5) (n5) {$\boldsymbol{\alpha_6:}(9,7)$\\ $a_1 = a_2$\\$x_1 = x_2$\\$y_1 = y_2$};
                
                \node[draw, rectangle, thick,align=center, rounded corners=3pt, fill=white,label={[xshift=0.7cm,yshift=-0.8mm]north:{\scriptsize$\{2\}, \{\alpha_8\}$}}] at (0,-5) (n6) {$\boldsymbol{\alpha_7:}(6,3)$\\ $a_1 = a_2$};
                
                \node[draw, rectangle, thick,align=center, rounded corners=3pt, fill=white,label={[xshift=-0.4cm]north:{\scriptsize$\{1,2\}$}}] at (2.5, -5) (n7) {$\boldsymbol{\alpha_8:}(6,4)$\\ $a_1 = a_2$\\$x_1 = x_2$};
                
                \node[draw, rectangle, thick,align=center, rounded corners=3pt, fill=white,label={[xshift=0.75cm,yshift=-0.8mm]north:{\scriptsize$\{1,2\}$, $\{\alpha_6\}$}}] at (5,-5) (n8) {$\boldsymbol{\alpha_9:}(8,6)$\\ $a_1 = a_2$\\$x_1 = x_2$\\$y_1 = y_2$};
                
                \draw[->, thick] (n2) + (1.3, 0) -- (n2);
                
                \draw[->, thick] (n2) -- (n1);
                
                \draw[->, thick] (n1) -- (n3);

                \draw[->, thick,transform canvas={xshift=0.9ex}] (n0) -- (n3);
                \draw[->, thick,transform canvas={xshift=-0.9ex}] (n3) -- (n0);
                
                \draw[->, thick] (n3) -- (n6);
                
                \draw[->, thick] (n6) -- (n7);

                 \draw[->, thick] (n7) -- (n4);
                 
                  \draw[->, thick] (n4) -- (n2);
                
                \draw[->, thick] (n8) -- (n5);
                
                \draw[->, thick] (n4) -- (n8);

                \draw[->, thick] (n5) -- (n4);
             
            \end{tikzpicture}
        }
        
        \subcaption{Winning strategy for the verifier.} \label{fig:winningStratEx}
    \end{subfigure}
    
    \caption{Two programs \texttt{Q1} and \texttt{Q2} are given in \Cref{fig:p1',fig:p2'}. In \Cref{fig:winningStratEx} a possible winning strategy for the verifier is depicted.  
        The outer label gives the scheduling $M \subseteq \{1, 2\}$ and, if applicable, the restriction chosen by the witness strategy. \vspace{-0.3cm}}\label{fig:exist}
\end{figure}
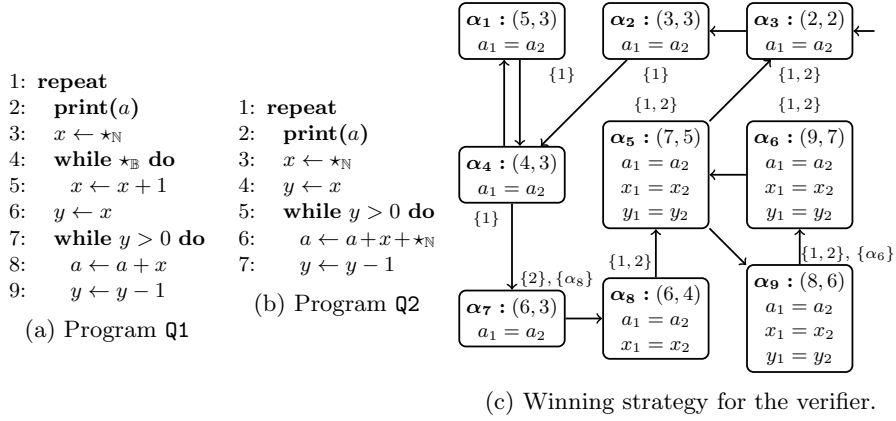

Our insight in this paper is that we can turn the strategy-based interpretation of the witness trace $\pi_2$ into a useful verification method by \emph{combining} it with a program reduction. 
As we express both searches strategically, we can phrase the combined search as a combined game. 
In particular, both the reduction and the witness strategy are controlled by the verifier and can thus \emph{collaborate}.
In the resulting game, the verifier chooses a scheduling (as in \Cref{sec:overview_ksafety}) and, additionally, whenever the existentially quantified copy is scheduled, the verifier also decides on the successor state of that copy. 
We depict a possible winning strategy in \Cref{fig:winningStratEx}.
This strategy formalizes the interplay of reduction and witness strategy.
Initially, the verifier only schedules $\{1\}$ until \texttt{Q1} has reached program location \texttt{Q1}:6 (at which point the value of $x$ is fixed).
Only then does the verifier schedule $\{2\}$, at which point the witness strategy can decide on a successor state for \texttt{Q2}.
In our case, the strategy chooses a value for $x$ such that $x_1 = x_2$ holds. 
As we work in an abstraction of the actual system, we formalize this by restricting the abstract successor states.
In particular, in state $\alpha_7$ the verifier schedules $\{2\}$ and simultaneously restricts the successors to $\{\alpha_8\}$ (i.e., the abstract state where $x_1 = x_2$ holds), even though abstract state $[(6, 4), a_1 = a_2, x_1 \neq x_2]$ is also a valid successors under scheduling $\{2\}$.
We formalize when a restriction is valid in \Cref{sec:beyondksafety}.
The resulting strategy is winning and therefore denotes both a reduction \emph{and} witness strategy for the existentially quantified copy.
Importantly, both reduction and witness strategy are mutually dependent. 
Our tool \texttt{HyPA} is able to verify both properties (in \Cref{fig:univ} and \Cref{fig:exist}) in a matter of a few seconds (cf.~\Cref{sec:eval}).

\section{Preliminaries}
\label{sec:prelim}

We begin by introducing basic preliminaries, including our basic model of computation and background on (finite-state) safety games.

\paragraph{Symbolic Transition Systems.}

We assume some fixed underlying first-order theory.
A \emph{symbolic transition system} (STS) is a tuple $\mathcal{T} = (X, \mathit{init}, \mathit{step})$ where $X$ is a finite set of variables (possibly sorted), $\mathit{init}$ is a formula over $X$ describing all initial states, and $\mathit{step}$ is a formula over $X \uplus X'$ (where $X' \coloneqq \{x' \mid x \in X\}$ is the set of primed variables) describing the transitions of the system.
A concrete state $\mu$ in $\mathcal{T}$ is an assignment to the variables in $X$.
We write $\mu'$ for the assignment over $X'$ given by $\mu'(x') \coloneqq  \mu(x)$.
A trace in $\mathcal{T}$ is an infinite sequence of assignment $\mu_0\mu_1\cdots$ such that $\mu_0 \models \mathit{init}$ and for every $i \in \nat$, $\mu_i \uplus \mu_{i+1}' \models \mathit{step}$.
We write $\traces{\STS}$ for the set of all traces in $\STS$.
We can naturally interpret programs as STS by making the program counter explicit. 

\paragraph{Formula Transformations.}

For the remainder of this paper, we fix the set of system variables $X$. 
We also fix a finite set of trace variables $\traceVars = \{\pi_1, \ldots, \pi_k\}$.
For a trace variable $\pi \in \traceVars$ we define $X_\pi \coloneqq  \{x_\pi \mid x \in X\}$ and write $\vec{X}$ for $X_{\pi_1} \cup \cdots \cup X_{\pi_k}$.
For a formula $\theta$ over $X$, we define $\theta_{\langle \pi \rangle}$ as the formula over $X_\pi$ obtained by replacing every variable $x$ with $x_\pi$.
Similarly, we define $k$ fresh disjoint copies $\vec{X}' = X'_{\pi_1} \cup \cdots \cup X'_{\pi_k}$ (where $X'_\pi \coloneqq  \{x'_\pi \mid x \in X\}$).
For a formula $\theta$ over $\vec{X}$, we define $\theta^{\langle'\rangle}$ as the formula over $\vec{X}'$ obtained by replacing every variable $x_\pi$ with $x_\pi'$.

\newcommand{\safe}{\texttt{SAFE}}
\newcommand{\reach}{\texttt{REACH}}

\paragraph{Safety Games.}

A \emph{safety game} is a tuple $\mathcal{G} = (S_\safe, S_\reach, S_0, T, B)$ where $S = S_\safe \uplus s_\reach$ is a set of game states, $S_0 \subseteq S$ a set of initial states, $T \subseteq S \times S$ a transition relation, and $B \subseteq S$ a set of bad states.
We assume that for every $s \in S$ there exists at least one $s'$ with $(s, s') \in T$.
States in $S_\safe$ are controlled by player $\safe$ and those in $S_\reach$ by player $\reach$.
A play is an infinite sequence of states $s_0s_1\cdots$ such that $s_0 \in  S_0$, and $(s_i, s_{i+1}) \in T$ for every $i \in \nat$.
A positional strategy $\sigma$ for player $p \in \{\safe, \reach\}$ is a function $\sigma : S_p \to S$ such that $(s, \sigma(s)) \in T$ for every $s \in S_p$.
A play $s_0s_1\cdots$ is compatible with strategy $\sigma$ for player $p$ if $s_{i+1} = \sigma(s_i)$ whenever $s_i \in S_p$.
The safety player wins $\mathcal{G}$ if there is a strategy $\sigma$ for $\safe$ such that all $\sigma$-compatible plays never visit a state in $B$.
In particular, $\safe$ needs to win from \emph{all} initial states.

\section{Observation-based HyperLTL}
\label{sec:ohyperltl}

In this section, we present OHyperLTL (short for observation-based HyperLTL). 
Our logic builds upon HyperLTL \cite{ClarksonFKMRS14}, which itself extends linear-time temporal logic (LTL) with explicit trace quantification.
In OHyperLTL, we include predicates from the background theory (to reason about infinite variable domains) and explicit observations (to express asynchronous properties).
Formulas in OHyperLTL are given by the following grammar:\footnote{For the examples in \Cref{sec:overview}, we additionally annotated quantifiers with an STS if we want to reason about different STSs within the same formula. 
	In the following, we assume that all quantifiers range over traces in the same STS to simplify notation.}
\begin{align*}
\varphi &\coloneqq  \forall \pi : \xi\ldot \varphi \mid \exists \pi : \xi\ldot \varphi \mid \phi \\
\phi &\coloneqq  \theta \mid  \neg \phi \mid \phi_1 \land \phi_2 \mid \ltlN \phi \mid\phi_1 \ltlU \phi_2
\end{align*}%
Here $\pi \in \traceVars$ is a trace variable, $\theta$ is a formula over $\vec{X}$, and $\xi$ is a formula over $X$ (called the observation formula).
For ease of notation, we assume that all variables in $\traceVars$ occur in the quantifier prefix \emph{exactly} once. 
We use the standard Boolean connectives $\wedge$, $\rightarrow$, $\leftrightarrow$, and constants $\top, \bot$, as well as the derived LTL operators eventually $\ltlF \phi \coloneqq \top \ltlU \phi $, and globally $\ltlG \phi \coloneqq \neg \ltlF \neg \phi$.

\newcommand{\filterdef}[2]{\mathit{valid}(#1,#2)}
\newcommand{\traceSet}{\mathbb{T}}

\paragraph{Semantics.}

A trace $t$ is an infinite sequence $\mu_0\mu_1\cdots$ of assignments to $X$.
For $i \in \nat$, we write $t(i)$ to denote the $i$th value in $t$.
A trace assignment $\Pi$ is a partial mapping of trace variables in $\traceVars$ to traces. 
Given a trace assignment $\Pi$ and $i \in \nat$, we define $\Pi(i)$ to be the assignment to $\vec{X}$ given by $\Pi(i)(x_\pi) \coloneqq  \Pi(\pi)(i)(x)$, i..e, the value of $x_\pi$ is the value of $x$ on the trace assigned to $\pi$.
For the LTL body of an OHyperLTL formula, we define:
\begin{align*}
	\Pi, i &\models  \theta &\text{ iff } \quad &\Pi(i) \models \theta \\
	\Pi, i &\models  \neg \phi &\text{ iff } \quad & \Pi, i \not\models  \phi \\
	\Pi, i &\models \phi_1 \land \phi_2 &\text{ iff } \quad  &\Pi, i \models \phi_1 \text{ and }  \Pi, i \models  \phi_2\\
	\Pi, i&\models  \ltlN  \phi &\text{ iff } \quad & \Pi , i + 1 \models \phi \\
	\Pi, i&\models  \phi_1 \ltlU \phi_2 &\text{ iff } \quad & \exists j \geq i\ldot \Pi, j\models  \phi_2 \text{ and } \forall i \leq k < j\ldot  \Pi, k \models  \phi_1
\end{align*}
The distinctive feature of OHyperLTL over HyperLTL are the explicit observations. 
Given an observation formula $\xi$ and trace $t$, we say that $\xi$ is a \emph{valid observation on} $t$ (written $\filterdef{t}{\xi}$) if there are infinitely many $i \in \nat$ such that $t(i) \models \xi$.
If $\filterdef{t}{\xi}$ holds, we write $\filter{t}{\xi}$ for the trace obtained by projecting on those positions $i$ where $t(i) \models \xi$, i.e., $\filter{t}{\xi}(i) \coloneqq  t(j)$ where $j$ is the $i$th index that satisfies $\xi$.
Given a set of traces $\traceSet$ we resolve trace quantification as follows:
\begin{align*}
	\Pi &\models_{\traceSet}  \phi &\text{ iff } \quad  &\Pi, 0 \models \phi \\
	\Pi &\models_{\traceSet}  \forall \pi : \xi\ldot \varphi &\text{ iff } \quad &\forall t \in \{t \in \traceSet \mid \filterdef{t}{\xi}\}\ldot  \Pi[\pi \mapsto \filter{t}{\xi}] \models_{\traceSet}  \varphi \\
	\Pi &\models_{\traceSet}  \exists \pi : \xi\ldot \varphi &\text{ iff } \quad &\exists t \in \{t \in \traceSet \mid \filterdef{t}{\xi}\}\ldot \Pi[\pi \mapsto \filter{t}{\xi}] \models_{\traceSet}  \varphi
\end{align*}
The semantics mostly agrees with that of HyperLTL\cite{ClarksonFKMRS14} but projects each trace to the positions where the observation holds. 
Given an STS $\mathcal{T}$ and OHyperLTL formula $\varphi$, we write $\mathcal{T} \models \varphi$ if $\emptyset \models_{\mathit{Traces}(\mathcal{T})} \varphi$ where $\emptyset$ is the empty assignment. 

\paragraph{The Power of Observations.}

The explicit observations in OHyperLTL facilitate the specification of asynchronous hyperproperties, i.e., properties where traces are traversed at different speeds.  
For the example in \Cref{sec:overview_ksafety}, the explicit observations allow us to compare the output of both programs even though the actual step at which the output occurs (in a synchronous semantics) differs between both programs (as \texttt{P1} takes the inner loop twice as often as \texttt{P2}). 
As the observations are part of the specification, we can model a broad spectrum of properties ranging, e.g., from timing-insensitive properties (by placing observations only at output locations) to timing-sensitive specifications \cite{GeYCH18} (by placing observations at closer intervals).
Functional (opposed to temporal) $k$-safety properties specified by pre-and postcondition \cite{Benton04,ShemerGSV19,UnnoTK21} can easily be encoded as $\forall^k$-OHyperLTL properties by placing observations at the start and end of each program. 
By setting $\xi = \top$, i.e., observing \emph{every} step, we can express synchronous properties.
OHyperLTL thus subsumes HyperLTL.

\paragraph{Finite-State Model Checking.}

Many mechanisms used to express asynchronous hyperproperties render finite-state model checking undecidable \cite{GutsfeldMO21,BozzelliPS21,BaumeisterCBFS21}.
In contrast, the simple mechanism used in OHyperLTL maintains decidable finite-state model checking. 
Detailed proofs can be found in \ifFull{the appendix}{the full version \cite{full}}.

\begin{restatable}{theorem}{finiteStateR}
Assume an STS $\STS$ with finite variable domains and decidable background theory and an OHyperLTL formula $\varphi$. It is decidable if $\mathcal{T} \models \varphi$.
\label{theo:finiteState}
\end{restatable}
\begin{proofSketch}
	Under the assumptions, we can view $\STS$ as an explicit (instead of symbolic) finite-state transition system. 
	Given an observation formula $\xi$ we can effectively compute an explicit finite-state system $\STS'$ such that $\traces{\STS'} = \{ \filter{t}{\xi} \mid t \in \traces{\STS} \land \filterdef{t}{\xi}\}$.
	This reduces OHyperLTL model checking on $\STS$ to HyperLTL model checking on $\STS'$, which is decidable \cite{FinkbeinerRS15}.
	\qed
\end{proofSketch}

Note that for infinite-state (symbolic) systems, we cannot effectively compute $\STS'$ as in the proof of \Cref{theo:finiteState}.
In fact, there may not even exist a system $\STS'$ with the desired property that is expressible in the same background theory.

The finite-state result in \Cref{theo:finiteState} is of little relevance for the present paper.
Nevertheless, it indicates that our logic is well suited for verification of infinite-state (software) systems as the (inevitable) undecidability stems from the infinite domains in software programs and not already from the logic itself.

\paragraph{Safety.}

In this paper, we assume that the hyperproperty is temporally safe \cite{BeutnerCFHK22}, i.e., the temporal body of any OHyperLTL formula denotes a \emph{safety property}.
Note that, as we support quantifier alternation, we can still express hyperliveness properties \cite{ClarksonS08,CoenenFST19}.
For example, GNI is both temporally safe and hyperliveness.
We model the body of a formula by a symbolic safety automaton \cite{DAntoniV17}, which is a tuple $\calA = (Q, q_0, \delta, B)$ where $Q$ is a finite set of states, $q_0 \in Q$ the initial state, $B \subseteq Q$ a set of bad-states, and $\delta$ a finite set of automaton edges of the form $(q, \theta, q')$ where $q, q' \in Q$ are states and $\theta$ is a formula over $\vec{X}$.
Given a trace $t$ over assignments to $\vec{X}$, a run of $\calA$ on $t$ is an infinite sequence of states $q_0q_1\cdots$ (starting in $q_0$) such that for every $i$, there exists an edge $(q_i, \theta_i, q_{i+1}) \in \delta$ such that $t(i) \models \theta_i$.
A word is accepted by $\calA$ if it has \emph{no} run that visits a state in $B$.
The automaton is \emph{deterministic} if for every $q \in Q$ and every assignments $\mu$ to $\vec{X}$, there exists exactly one edge $(q, \theta, q') \in \delta$ with $\mu \models \theta$.

\section{Reductions as a Game}
\label{sec:ksafety}

After having defined our temporal logic, we turn our attention to the automatic verification of OHyperLTL formulas on STSs. 
In this section, we begin by formalizing our game-based interpretation of a reduction.
To illustrate this, we consider $\forall^k$ OHyperLTL formulas, which, as the body of the formula is a safety property, always denote $k$-safety properties. 

\paragraph{Predicate Abstraction.}

Our search for a reduction is based in the scope of a fixed predicate abstraction \cite{GrafS97,JhalaPR18}, i.e., we abstract our system by keeping track of the truth value of a few selected predicates that (ideally) identify properties that are relevant to prove the property in question.
Let $\mathcal{T} = (X, \mathit{init}, \mathit{step})$ be an STS and let $\varphi = \forall \pi_1 : \xi_1 \ldots \forall \pi_k : \xi_k\ldot \phi$ be the ($k$-safety) OHyperLTL we wish to verify.
Let $\calA_\phi = (Q_\phi, q_{\phi, 0}, \delta_\phi, B_\phi)$ be a deterministic safety automaton for $\phi$.
A \emph{relational} predicate $p$ is a formula over $\vec{X}$ that identifies a property of the combined state space of $k$ system copies.
Let $\mathcal{P} = \{p_1, \ldots, p_n\}$ be a finite set of relational predicates.
We say a formula over $\vec{X}$ is \emph{expressible in $\mathcal{P}$} if it is equivalent to a boolean combination of the predicates in $\mathcal{P}$.
We assume that all edge formulas in the automaton $\calA_\phi$, and formulas $\mathit{init}_{\langle \pi_i \rangle}$ and $(\xi_{i})_{\langle \pi_i \rangle}$ for $\pi_i \in \traceVars$ are expressible in $\mathcal{P}$.
Note that we can always add missing predicates to $\mathcal{P}$.

Given the set of predicates $\pred$, the state-space of the abstraction w.r.t.~$\mathcal{P}$ is given by $\mathbb{B}^n$, where for each abstract state $\hat{s} \in \mathbb{B}^n$, the $i$th position $\hat{s}[i] \in \mathbb{B}$ tracks whether or not predicate $p_i$ holds.
To simplify notation, we write $\mathit{ite}(b, \theta, \theta')$ to be formula $\theta$ if $b = \top$, and $\theta'$ otherwise. 
For each abstract state $\hat{s} \in \mathbb{B}^n$, we define $\predSat{\hat{s}} \coloneqq  \bigwedge_{i = 1}^n  \mathit{ite}\big(\hat{s}[i] , p_i, \neg p_i\big)$, i.e., $\predSat{\hat{s}}$ is a formula over $\vec{X}$ that captures all concrete states that are abstracted to $\hat{s}$.
To incorporate reductions in our abstraction, we parametrize the abstract transition relation by a \emph{scheduling} $M \subseteq \{\pi_1, \ldots, \pi_k\}$.
We lift the $\mathit{step}$ formula from $\STS$ by defining
\begin{align*}
	\mathit{step}_M \coloneqq   \bigwedge_{i = 1}^k \mathit{ite}\Big(\pi_i \in M, \mathit{step}_{\langle \pi_i \rangle}, \bigwedge_{x \in X} x_{\pi_i}' = x_{\pi_i}\Big).
\end{align*} 
That is all copies in $M$ take a step while all other copies remain unchanged.
Given two abstract states $\hat{s}_1, \hat{s}_2$ we say that $\hat{s}_2$ is an \emph{$M$-successor} of $\hat{s}_1$, written $\hat{s}_1 \xrightarrow{M} \hat{s}_2$, if $\predSat{\hat{s}_1} \land \predSat{\hat{s}_2}^{\langle'\rangle} \land \mathit{step}_M$ is satisfiable, i.e., we can transition from $\hat{s}_1$ to $\hat{s}_2$ by only progressing the copies in $M$.

For an abstract state $\hat{s}$, we define $\mathit{obs}(\hat{s}) \in \mathbb{B}^k$ as the boolean vector that indicates which copy (of $\pi_1, \ldots, \pi_k$) is currently at an observation point, i.e., $\mathit{obs}(\hat{s})[i] = \top$ iff $\predSat{\hat{s}} \land (\xi_{i})_{\langle \pi_i \rangle}$ is satisfiable.
Note that as $(\xi_{i})_{\langle \pi_i \rangle}$ is, by assumption, expressible in $\mathcal{P}$, either all or none of the concrete states in $\predSat{\hat{s}}$ satisfy $(\xi_{i})_{\langle \pi_i \rangle}$.

\paragraph{Game Construction.}

Building on the parametrized abstract transition relation, we can construct a (finite-state) safety game where winning strategies for the verifier correspond to valid reductions with accompanying proofs.
The nodes in our game have two forms:
Either they are of the form $(\hat{s}, q, b)$ where $\hat{s} \in \mathbb{B}^n$ is an abstract state, $q \in Q_\phi$ a state of the safety automaton, and $b \in \mathbb{B}^k$ a boolean vector indicating which copy has moved since the last automaton step;
Or of the form $(\hat{s}, q, b, M)$ where $\hat{s}$, $q$, and $b$ are as before and $\emptyset \neq M \subseteq  \{\pi_1, \ldots, \pi_k\}$ is a scheduling.
The initial states are all states $(\hat{s},q_{\phi, 0}, \top^k)$ where $\predSat{\hat{s}} \land  \bigwedge_{i=1}^k \mathit{init}_{\langle \pi_i \rangle}$ is satisfiable (recall that $\mathit{init}_{\langle \pi_i \rangle}$ is expressible in $\pred$).
We mark a state $(\hat{s}, q, b)$ or $(\hat{s}, q, b, M)$ as losing iff $q \in B_\phi$.
For automaton state $q \in Q_\phi$ and abstract state $\hat{s}$, we define $\delta_\phi(q, \hat{s})$ as the \emph{unique} state $q'$ such that there is an edge $(q, \theta, q') \in \delta_\phi$ such that $\predSat{\hat{s} }\land \theta$ is satisfiable.
Uniqueness follows from the assumption that $\calA_\phi$ is deterministic and all edge formulas are expressible in $\mathcal{P}$.
The transition relation of our game is given by the following rules:

\noindent
\begin{minipage}{0.5\textwidth}
	\def\defaultHypSeparation{\hskip .1in}
	\def\ScoreOverhang{0pt}
	\scalebox{1}{
		\begin{minipage}{\textwidth}
			\begin{prooftree}
				\AxiomC{$\forall \pi_i \in M\ldot \neg b[i]  \lor \neg \mathit{obs}(\hat{s})[i]$}
				\RightLabel{\footnotesize\textbf{(1)}}
				\UnaryInfC{$(\hat{s}, q, b) \rightsquigarrow (\hat{s}, q, b, M)$}
			\end{prooftree}
		\end{minipage}
	}
\end{minipage}%
\begin{minipage}{0.5\textwidth}
	\def\defaultHypSeparation{\hskip .1in}
	\def\ScoreOverhang{0pt}
	\scalebox{1}{
		\begin{minipage}{\textwidth}
			\begin{prooftree}
				\AxiomC{$\mathit{obs}(\hat{s}) = \top^k$}
				\AxiomC{$q' = \delta_\phi(q, \hat{s})$}
				\RightLabel{\footnotesize\textbf{(2)}}
				\BinaryInfC{$(\hat{s}, q, \top^k) \rightsquigarrow (\hat{s}, q', \bot^k)$}
			\end{prooftree}
		\end{minipage}
	}
\end{minipage}%

	{\def\defaultHypSeparation{\hskip .1in}
	\def\ScoreOverhang{0pt}
	\scalebox{1}{
		\begin{minipage}{\textwidth}
			\begin{prooftree}
				\vspace{2mm}
				\AxiomC{$\hat{s} \xrightarrow{M} \hat{s}'$}
				\AxiomC{$b' = b[i \mapsto \top]_{\pi_i \in M}$}
				\RightLabel{\footnotesize\textbf{(3)}}
				\BinaryInfC{$(\hat{s}, q, b, M) \rightsquigarrow (\hat{s}', q, b')$}
			\end{prooftree}
		\end{minipage}
	}}

\vspace{0mm}
\noindent
In rule \textbf{(1)}, we select any scheduling that schedules only copies that have not reached an observation point or have not moved since the last automaton step. 
In particular, we cannot schedule any copy that has moved and already reached an observation point.
In rule \textbf{(2)}, all copies reached an observation point and have moved since the last update (i.e., $b = \top^k$) so we progress the automaton and reset $b$.
Lastly, in rule \textbf{(3)}, we select an $M$-successor of $\hat{s}$ and update $b$ for all copies that take part in the step.
In our game, player $\safe$ takes the role of the verifier, and player $\reach$ that of the refuter. 
It is the safety player's responsibility to select a scheduling in each step, so we assign nodes of the form $(\hat{s}, q, b)$ to $\safe$.
Nodes of the form $(\hat{s}, q, b, M)$ are controlled by $\reach$ who can choose an abstract $M$-successor.
Let $\gameE{\STS}{\varphi}{\pred}$ be the resulting (finite-state) safety game.
A winning strategy for $\safe$ in $\gameE{\STS}{\varphi}{\pred}$ picks, in each abstract state, a valid scheduling that prevents a visit to a losing state.
We can thus show:

\begin{restatable}{theorem}{soundnessksafety}\label{theo:soundessKSafety}
If player $\safe$ wins $\gameE{\STS}{\varphi}{\pred}$, then $\STS \models \varphi$.
\end{restatable}
\begin{proofSketch}
	Assume $\sigma$ is a winning strategy for $\safe$ in $\gameE{\STS}{\varphi}{\pred}$.
	Let $t_1, \ldots, t_k \in \traces{\STS}$ be arbitrary.
	We, iteratively, construct stuttered versions $t'_1, \ldots, t'_k$ of $t_1, \ldots, t_k$ by querying $\sigma$ on abstracted prefixes of $t_1, \ldots, t_k$:
	Whenever $\sigma$ schedules copy $i$ we take a proper step on $t_i$; otherwise we stutter.
	By construction of $\gameE{\STS}{\varphi}{\pred}$ the stuttered traces $t'_1, \ldots, t'_k$  align at observation points.
	In particular, we have $[\pi_1 \mapsto \filter{t_1}{\xi_1}, \ldots, \pi_k \mapsto \filter{t_k}{\xi_k}] \models \phi$ iff $[\pi_1 \mapsto \filter{t'_1}{\xi_1}, \ldots, \pi_k \mapsto \filter{t'_k}{\xi_k}] \models \phi$.
	Moreover, the sequence of abstract states in $\gameE{\STS}{\varphi}{\pred}$ forms an abstraction of $t'_1, \ldots, t'_k$ and shows that $\calA_\phi$ cannot reach a bad state when reading $\filter{t'_1}{\xi_1}, \ldots, \filter{t'_k}{\xi_k}$ (as $\sigma$ is winning).
	This already shows that $[\pi_1 \mapsto \filter{t'_1}{\xi_1}, \ldots, \pi_k \mapsto \filter{t'_k}{\xi_k}] \models \phi$ and thus $[\pi_1 \mapsto \filter{t_1}{\xi_1}, \ldots, \pi_k \mapsto \filter{t_k}{\xi_k}] \models \phi$.
	As this holds for all traces  $t_1, \ldots, t_k \in \traces{\STS}$, we get $\STS \models \varphi$ as required. 
	\qed
\end{proofSketch}

\paragraph{Game Construction and Complexity.}

If the background theory is decidable, $\gameE{\STS}{\varphi}{\pred}$  can be constructed effectively using at most $2^{|\mathcal{P}|+1} \cdot 2^k$ queries to an SMT solver.
Checking if $\safe$ wins $\gameE{\STS}{\varphi}{\pred}$ can be done with a simple fixpoint computation of the attractor in linear time.

Our game-based method of finding a reduction in a given abstraction is closely related to the notation of a \emph{property-directed self-composition} \cite{ShemerGSV19}.
The previously only known algorithm for finding such a reduction is based on an optimized enumeration \cite{ShemerGSV19}, which, in the worst case, requires $\mathcal{O}(2^{|\mathcal{P}|+1} \cdot 2^k)$ many enumerations.
Our worst-case complexity thus matches the bounds inferred by \cite{ShemerGSV19}, but avoids the explicit enumeration of reductions (and the concomitant repeated construction of the abstract state-space) and is, as we believe, conceptually simpler to comprehend.
Moreover, our game-based technique is the key stepping stone for extending our method beyond $k$-safety in \Cref{sec:beyondksafety}. 

\section{Verification Beyond $k$-Safety}
\label{sec:beyondksafety}

Building on the game-based interpretation of a reduction, we extend our verification beyond $\forall^*$ properties to support $\forall^*\exists^*$ properties.
We accomplish this by \emph{combining} the game-based reading of a reduction (as discussed in the previous section) with a game-based reading of existential quantification.
For the remainder of this section, fix an STS $\calT = (X, \mathit{init}, \mathit{step})$ and let 
\begin{align*}
	\varphi = \forall \pi_1 : {\xi_1} \ldots \forall \pi_l : {\xi_l}. \exists \pi_{l+1} : {\xi_{l+1}} \ldots \exists \pi_k : {\xi_{k}}\ldot \phi
\end{align*}
be the OHyperLTL formula we wish to check, i.e., we universally quantify over $l$ traces followed by an existential quantification over $k-l$ traces.
We assume that for every existential quantification $\exists \pi_i : \xi_i$ occurring in $\varphi$, $\filterdef{t}{\xi_i}$ holds for every $t \in \traces{\STS}$ (we discuss this later in \Cref{rem:assumption}).

\subsection{Existential Trace Quantification as a Game}

The idea of a game-based verification of $\forall^*\exists^*$ properties is to consider a $\forall^*\exists^*$-property as a game between verifier and refuter \cite{CoenenFST19}. 
The refuter controls the $l$ universally quantified traces by moving through $l$ copies of the system (thereby producing traces $\pi_1, \ldots, \pi_l$) and the verifier reacts by, incrementally, moving through $k-l$ copies of the system (thereby producing traces $\pi_{l+1}, \ldots, \pi_k$). 
If the verifier has a strategy that ensures that the resulting traces satisfy $\phi$, $\STS \models \varphi$ holds. 
We call such a strategy for the verifier a \emph{witness strategy}.

We combine this game-based reading of existential quantification with our game-based interpretation of a reduction by, additionally, letting the verifier control the scheduling of the system. 
When played on the \emph{concrete} state-space of $\STS$ the game proceeds in three stages as follows:
1) The verifier selects a valid scheduling $M \subseteq \{\pi_1, \ldots, \pi_k\}$;
2) The refuter selects successor states for all universally quantified copies by fixing an assignment to $X_{\pi_1}', \ldots, X_{\pi_l}'$ (only moving copies scheduled by $M$);
3)  The verifier reacts by choosing successor states for the existentially quantified copies by fixing an assignment to $X_{\pi_{l+1}}', \ldots, X_{\pi_k}'$ (again, only moving copies scheduled by $M$).
Afterward, the process repeats.

As we work within a fixed abstraction of $\STS$, the verifier can, however, not choose concrete successor states directly but only work in the precision captured by the abstraction.
Following the general scheme of abstract games, we, therefore, underapproximate the moves available to the verifier \cite{AlfaroGJ04}.
Formally, we abstract the three-stage game outlined before (which was played at the level of concrete states) to a simpler abstract game (consisting of only two stages).
In the first stage, the verifier selects both a scheduling $M$ and a \emph{restriction} on the set of abstract successor states, i.e., a set of abstract states $A$.
In the second stage, the refuter cannot choose any abstract successor state (any $M$-successor in the terminology from \Cref{sec:ksafety}), but only successors contained in the restriction $A$.
To guarantee the soundness of this approach, we ensure that the verifier can only pick restrictions that are \emph{valid}, i.e., restrictions that underapproximate the possibilities of the verifier on the level of concrete states.

\paragraph{Game Construction.}
We modify our game from \Cref{sec:ksafety} as follows.
States are either of the form $(\hat{s}, q, b)$ (as in \Cref{sec:ksafety}) or of the form $(\hat{s}, q, b, M, A)$ where $\hat{s}$, $q$, $b$, and $M$ are as in \Cref{sec:ksafety}, and $A \subseteq \mathbb{B}^n$ is a subset of abstract states (the restriction).
To reflect the restriction, we modify transition rules \textbf{(1)} and \textbf{(3)}. Rule \textbf{(2)} remains unchanged.

\noindent
\begin{minipage}{0.55\textwidth}
    \def\defaultHypSeparation{\hskip .2in}
    \def\ScoreOverhang{0pt}
    \scalebox{1}{
        \begin{minipage}{\textwidth}
             \begin{prooftree}
                \AxiomC{$\forall \pi_i \in M\ldot \neg b[i]  \lor \neg \mathit{obs}(\hat{s})[i]$}
                \AxiomC{$\validRes{\hat{s}}{M}{A}$}
                \RightLabel{\footnotesize\textbf{(1)}}
                \BinaryInfC{$(\hat{s}, q, b) \rightsquigarrow (\hat{s}, q, b, M, A)$}
            \end{prooftree}
        \end{minipage}
    }
\end{minipage}%
\begin{minipage}{0.45\textwidth}
    \def\defaultHypSeparation{\hskip .1in}
    \def\ScoreOverhang{0pt}
    \scalebox{1}{
        \begin{minipage}{\textwidth}
            \begin{prooftree}
                \AxiomC{$\hat{s}' \in A$}
                \AxiomC{$b' = b[i \mapsto \top]_{i \in M}$}
                \RightLabel{\footnotesize\textbf{(3)}}
                \BinaryInfC{$(\hat{s}, q, b, M, A) \rightsquigarrow (\hat{s}', q, b')$}
            \end{prooftree}
        \end{minipage}
    }
\end{minipage}

\vspace{2mm}

\noindent
In rule \textbf{(1)}, the safety player (who, again, takes the role of the verifier) selects both a scheduling $M$ and a restriction $A$ such that $\validRes{\hat{s}}{M}{A}$ holds (which we define later).
The reachability player (who takes the role of the refuter) can, in rule \textbf{(3)}, select any successor contained in $A$.

\paragraph{Valid Restriction.}

The above game construction depends on the definition of $\validRes{\hat{s}}{M}{A}$.
Intuitively, $A$ is a valid restriction if it underapproximates the possibilities of a witness strategy that can pick concrete successor states for all existentially quantified traces.
That is, for every concrete state in $\hat{s}$, a witness strategy (on the level of concrete states) can guarantee a move to a concrete state that is abstracted to an abstract state within $A$.
Formally we define $\validRes{\hat{s}}{M}{A}$ as follows:
\vspace{-1mm}
\begin{align*}
        &\forall \{X_{\pi_i}\}_{i=1}^{k}. \forall \{X_{\pi_i}'\}_{i=1}^l. \; \predSat{\hat{s}} \land \bigwedge_{i=1}^l \mathit{ite}\Big(\pi_i \in M, \mathit{step}_{\langle \pi_i \rangle}, \bigwedge_{x \in X} x_{\pi_i}' = x_{\pi_i}\Big)\\[-0.3cm]
        &\quad\Rightarrow \exists \{X_{\pi_i}'\}_{i=l+1}^k. \bigwedge_{i=l+1}^k \mathit{ite}\Big(\pi_i \in M, \mathit{step}_{\langle \pi_i \rangle}, \bigwedge_{x \in X} x_{\pi_i}' = x_{\pi_i}\Big) \land \bigvee\limits_{\hat{s}' \in A} \predSat{\hat{s}'}^{\langle'\rangle}
\end{align*}
It expresses that for all concrete states in $\predSat{\hat{s}}$ (assignments to $\{X_{\pi_i}\}_{i=1}^{k}$) and for all concrete successor states for the universally quantified copies (assignments to $\{X_{\pi_i}'\}_{i=1}^l$), there exist successor states for the existentially quantified copies ($\{X_{\pi_i}'\}_{i=l+1}^k$) such that one of the abstract states in $A$ is reached.

\begin{example}
	With this definition at hand, we can validate the restrictions chosen by the strategy in \Cref{fig:winningStratEx}.
	For example, in state $\alpha_7$ the strategy schedules $M = \{2\}$ and restricts the successor states to $\{\alpha_8\}$ even though abstract state $\big[(6, 4), a_1 = a_2, x_1 \neq x_2\big]$ is also a $\{2\}$-successor of $\alpha_7$.
	If we spell out $\validRes{\alpha_7}{\{2\}}{\{\alpha_8\}}$ we get\\[-2mm]
	\scalebox{0.9}{\parbox{\linewidth}{
	\begin{align*}
		&\forall X_1 \!\cup\! X_2 \!\cup\! X_1'\ldot \; \underbrace{a_1 = a_2}_{\predSat{\alpha_7}} \land  \Big(\!\bigwedge_{z \in X} \! z'_1 = z_1\Big) \Rightarrow \exists X_2'\ldot \; \underbrace{a_2' = a_2 \land y_2' = y_2}_{\mathit{step_{\langle 2 \rangle}}} \land \underbrace{\big(a_1' = a_2' \land x_1' = x_2'\big)}_{\predSat{\alpha_8}^{\langle'\rangle}}
	\end{align*}
}}\\[-2mm]
	where $X = \{a, x, y\}$.
	Here we assume that $\mathit{step} \coloneqq  \big(a' = a \land y' = y\big)$ is the update performed on instruction $x \leftarrow \star_\nat$ from \texttt{Q2}:3 to \texttt{Q2}:4. 
	The above formula is valid.
\end{example}

\paragraph{Correctness.}

Call the resulting game $\gameEF{\STS}{\varphi}{\pred}$.
The game combines the search for a reduction with that of a witness strategy (both within the precision captured by $\pred$).\footnote{In particular, $\gameEF{\STS}{\varphi}{\pred}$ (strictly) generalizes the construction of $\gameE{\STS}{\varphi}{\pred}$ from \Cref{sec:ksafety}:
	If $k = l$ (i.e, the property is a $\forall^*$-property) the unique minimal valid restriction from $\hat{s}, M$ is $\{\hat{s}' \mid \hat{s} \xrightarrow{M} \hat{s}'\}$, i.e., the set of all $M$-successors of $\hat{s}$.
	The safety player can thus not be more restrictive than allowing \emph{all} $M$-successors (as in $\gameE{\STS}{\varphi}{\pred}$).}
We can show:

\begin{restatable}{theorem}{soundessBeyond}
If player $\safe$ wins $\gameEF{\STS}{\varphi}{\pred}$, then $\mathcal{T} \models \varphi$.
\label{theo:soundnessBeyond}
\end{restatable}
\begin{proofSketch}
	Let $\sigma$ be a winning strategy for $\safe$ in $\gameEF{\STS}{\varphi}{\pred}$.
	Let $t_1, \ldots, t_l \in \traces{\STS}$ be arbitrary.
	We use $\sigma$ to incrementally construct witness traces $t_{l+1}, \ldots, t_k$ by querying $\sigma$.
	In every abstract state $\hat{s}$, $\sigma$ selects a scheduling $M$ and a restriction $A$ such that $\validRes{\hat{s}}{M}{A}$ holds.
	We plug the current \emph{concrete} state (reached in our construction of $t_{l+1}, \ldots, t_k$) into the universal quantification of $\validRes{\hat{s}}{M}{A}$ and get (concrete) witnesses for the existential quantification that, by definition of $\validRes{\hat{s}}{M}{A}$, are valid successors for the existentially quantified copies in $\STS$.
	\qed
\end{proofSketch}

\begin{remark}\label{rem:assumption}
	Recall that we assume that for every existential quantification $\exists \pi_i : \xi_i$ occurring in $\varphi$ and all $t \in \traces{\STS}$, $\filterdef{t}{\xi_i}$ holds.
	This is important to ensure that the safety player (the verifier) cannot avoid observation points forever. 
	We could drop this assumption by strengthening the winning condition in $\gameEF{\STS}{\varphi}{\pred}$ and explicitly state that, in order to win, $\safe$ needs to visit observations points on existentially quantified traces infinitely many times. 
\end{remark}

\paragraph{Clairvoyance vs.~Abstraction.}

The cooperation between reduction (the ability of the verifier to select schedulings) and witness strategy (the ability to select restrictions on the successor) can be seen as a limited form of prophecy \cite{AbadiL91,BeutnerF22}.
By first scheduling the universal copies, the witness strategy can peek at future moves before committing to a successor state, as we e.g., saw in \Cref{fig:exist}.
The ``theoretically optimal'' reduction is thus a sequential one that first schedules only the universally quantified traces (until an observation point is reached) and thereby provides maximal information for the witness strategy.
However, in the context of a fixed abstraction, this reduction is not always optimal. 
For example, in \Cref{fig:exist} the strategy schedules the loop in lock-step which is crucial for generating a proof with simple (linear) invariants.
In particular, \Cref{fig:exist} does not admit a witness strategy in the lock-step reduction and does not admit a proof with linear invariants in a sequential reduction.
Our verification framework, therefore, strikes a delicate balance between clairvoyance needed by the witness strategy and precision captured in the abstraction, further emphasizing why the searches for reduction and witness strategy need to be mutually dependent. 

\subsection{Constructing and Solving $\gameEF{\STS}{\varphi}{\pred}$}\label{sec:solvingInner}

\begin{wrapfigure}{R}{0.57\textwidth}
    \vspace{-0.5cm}
    \scalebox{1}{
        \begin{minipage}{0.55\textwidth}
            \begin{algorithm}[H]
                \caption{Iterative solver for $\gameEF{\STS}{\varphi}{\pred}$. 
            }\label{alg:cap}
                \begin{algorithmic}[1]
                    \State \textbf{Input: } $\mathcal{T}, \varphi, \mathcal{P}$                    
                    \State $\tilde{\mathcal{G}}\coloneqq  \mathit{initialApproximation}(\mathcal{T}, \varphi, \mathcal{P})$
                    \RepeatOne
                   
                    \MatchTwo{$\mathit{Solve}(\tilde{\mathcal{G}})$}
                    
                    \State \textbf{case }$\texttt{REACH}(\sigma)$\textbf{:} \textbf{return} $\texttt{REACH}$
                    \CaseTwo{$\texttt{SAFE}(\sigma)$}
                    
                    \ForAllOne{$(\hat{s}, M, A) \in \mathit{Restrictions}(\sigma)$}
                    \IfTwo{$\neg \validRes{\hat{s}}{M}{A}$}
                    \ForAllOne{$A' \subseteq A$ }
                    \State $\tilde{\mathcal{G}}\coloneqq  \mathit{Remove}(\tilde{\mathcal{G}}, (\hat{s}, M, A'))$
                    \State \textbf{goto } 4
                    
                    \State \textbf{return} $\texttt{SAFE}$
                \end{algorithmic}
            \end{algorithm}
        \end{minipage}
    }
    \vspace{-0.2cm}
    
\end{wrapfigure}

Constructing the game graph of $\gameEF{\STS}{\varphi}{\pred}$ requires the identification of all valid restrictions (of which there are exponentially many in the number of abstract states and thus double exponentially many in the number of predicates) each of which requires to solve a quantified SMT query.
We propose a more effective algorithm that solves $\gameEF{\STS}{\varphi}{\pred}$ without constructing it explicitly.
Instead, we iteratively refine an abstraction $\tilde{\mathcal{G}}$ of $\gameEF{\STS}{\varphi}{\pred}$.
Our method hinges on the following easy observation:

\begin{lemma}\label{lem:closed}
	For any $\hat{s}$ and $M$, $\{A \mid \validRes{\hat{s}}{M}{A}\}$ is upwards closed (w.r.t.~$\subseteq$).
\end{lemma} 

Our initial abstraction consists of all possible restrictions (even those that might be invalid), i.e., we allow all restrictions of the form $(\hat{s}, M, A)$ where $A \subseteq \{\hat{s}' \mid \hat{s} \xrightarrow{M} \hat{s}'\}$.\footnote{Note that $\{\hat{s}' \mid \hat{s} \xrightarrow{M} \hat{s}'\}$ is always a valid restriction. Importantly, we can compute $ \{\hat{s}' \mid \hat{s} \xrightarrow{M} \hat{s}'\}$ locally, i.e., by iterating over abstract states opposed to \emph{sets} of abstract states.  }
This overapproximates the power of the safety player, i.e., a winning strategy for $\safe$ in $ \tilde{\mathcal{G}}$ may not be valid in $\gameEF{\STS}{\varphi}{\pred}$.
To remedy this, we propose the following inner refinement loop:
If we find a winning strategy $\sigma$ for $\safe$ in $\tilde{\mathcal{G}}$ we check if all restrictions chosen by $\sigma$ are valid.
If this is the case, $\sigma$ is also winning for $\gameEF{\STS}{\varphi}{\pred}$ and we can apply \Cref{theo:soundnessBeyond}.
If we find an invalid restriction $(\hat{s}, M, A)$ used by $\sigma$, we refine $\tilde{\mathcal{G}}$ by removing not only the restriction $(\hat{s}, M, A)$  but \emph{all} $(\hat{s}, M, A')$ with $A' \subseteq A$ (which is justified by \Cref{lem:closed}).
The algorithm is sketched in \Cref{alg:cap}.
The subroutine $\mathit{Restrictions}(\sigma)$ returns all restrictions used by $\sigma$, i.e., all tuples $(\hat{s}, M, A)$ such that $\sigma$ uses an edge $(\hat{s}, q, b) \rightsquigarrow (\hat{s}, q, b, M, A)$ for some $q, b$.
$\mathit{Remove}\linebreak[1](\tilde{\mathcal{G}},\linebreak[1] (\hat{s}, M, A'))$ removes from $\tilde{\mathcal{G}}$ all edges of the form $(\hat{s}, q, b) \rightsquigarrow (\hat{s}, q, b, M, A')$ for some $q, b$, and $\mathit{Solve}$ solves a finite-state safety game. 
To improve the algorithm further, in line 4 we always compute a maximal safety strategy, i.e., a strategy that selects maximal restrictions (w.r.t.~$\subseteq$) and therefore allows us to eliminate many invalid restrictions from $\tilde{\mathcal{G}}$ simultaneously.
For safety games, there always exists such a maximal winning strategy (see e.g.~\cite{BernetJW02}).
Note that while $\tilde{\mathcal{G}}$ is large, solving this finite-state game can be done very efficiently.
The running time of solving $\gameEF{\STS}{\varphi}{\pred}$ is dominated by the SMT queries of which our refinement loop, in practice, requires very few.

\section{Implementation and Evaluation}
\label{sec:eval}

\begin{wraptable}{R}{0.51\textwidth}
	\vspace{-3mm}
	\captionsetup{width=\linewidth-4mm}
	\caption{Evaluation of \texttt{HyPA} on $k$-safety instances. We give the size of the abstract game-space (Size), the time taken to compute the abstraction ($t_\mathit{abs}$), and the overall time taken by \texttt{HyPA} ($t$).
		Times are given in seconds.
	} \label{tab:ksafetyeval}
	
	\vspace{-1mm}
	\centering
	
	\small
	
	\def\arraystretch{1.3}
	\setlength\tabcolsep{1.7mm}
	\setlength\dashlinedash{1pt}
	\setlength\dashlinegap{2pt}
	\setlength\arrayrulewidth{0.7pt}
	\begin{tabular}{lccc}
		\toprule[1pt]
		\textbf{Instance} & \textbf{Size} & $\boldsymbol{t}_\mathit{abs}$ & $\boldsymbol{t}$  \\
		\cmidrule[1pt](r{1mm}){1-1}
		\cmidrule[1pt](l{1mm}r{1mm}){2-2}
		\cmidrule[1pt](l{1mm}r{1mm}){3-3}
		\cmidrule[1pt](l{1mm}r{0mm}){4-4}
		DoubleSquareNI & 819  & 92.3 & 92.8  \\
		\hdashline
		HalfSquareNI & 1166  & 85.9 & 86.5  \\
		\hdashline
		SquaresSum &  286 & 29.8  & 29.9 \\
		\hdashline
		ArrayInsert & 213 & 28.2 & 28.2 \\
		\hdashline
		Exp1x3 &  112 & 4.5  & 4.5 \\
		\hdashline
		Fig3& 268 & 11.9  & 12.0 \\
		\hdashline
		DoubleSquareNIff & 121 & 9.8 & 9.9 \\
		\hdashline
		\Cref{fig:univ} & 333 & 23.7 & 23.8 \\
		\hdashline
		ColIitem-Symm & 494 & 24.0 & 24.1 \\
		\hdashline
		Counter-Det & 216 & 10.2 & 10.3 \\
		\hdashline
		MultEquiv  & 757 & 18.9 & 19.0 \\
		\bottomrule[1pt]
	\end{tabular}
	\vspace{-3mm}
\end{wraptable}

When combining \Cref{theo:soundnessBeyond} and our iterative solver from \Cref{sec:solvingInner} we obtain an algorithm to verify $\forall^*\exists^*$-safety properties within a given abstraction. 
We have implemented a prototype of our method in a tool we call \texttt{HyPA}.
We use \texttt{Z3} \cite{MouraB08} to discharge SMT queries. 
The input of our tool is provided as an arbitrary STS in the SMTLIB format \cite{barrett2010smt}, making it \emph{language independent}.
In our programs, we make the program counter explicit, allowing us to track predicates locally \cite{HenzingerJMS02}.

\paragraph{Evaluation for $k$-Safety.}

As a special case of $\forall^*\exists^*$ properties, \texttt{HyPA} is also applicable to $k$-safety verification.
We collected an exemplifying suite of programs and $k$-safety properties from the literature \cite{ShemerGSV19,FarzanV19,UnnoTK21,SousaD16,UnnoTK21} and manually translated them into STS (this can be automated easily).
The results are given in \Cref{tab:ksafetyeval}.
As done by Shemer et al.~\cite{ShemerGSV19}, we already provide a set of predicates that is sufficient for \emph{some} reduction (but not necessarily the lockstep or sequential one), the search for which is then automated by \texttt{HyPA}.
Our results show the game-based search for a reduction can verify interesting $k$-safety properties from the literature.
We also note that, currently, the vast majority of time is spent on the construction of the abstract system.
If we would move to a fixed language, the computation time of the initial abstraction could be reduced by using existing (heavily optimized) abstraction tools \cite{HenzingerJMS02,ChakiCGJV04}.

\begin{table}[!t]
	\caption{Evaluation of \texttt{HyPA} on $\forall^*\exists^*$-safety verification instances. 
		We give the size and construction time of the initial abstraction (Size and $t_\mathit{abs}$).
		For both the direct (explicit) and lazy (\Cref{alg:cap}) solver we give the time to construct (and solve) the game ($t_\mathit{solve}$) and the overall time ($t = t_\mathit{abs} + t_\mathit{solve}$).
		For the lazy solver we, additionally, give the number of refinement iterations (\#Ref).
		Times are given in seconds. 
		TO indicates a timeout after 5 minutes.} \label{tab:beyondeval}
	
	\vspace{2mm}
	
	\centering
	
	\small
	
	\def\arraystretch{1.3}
	\setlength\tabcolsep{1.7mm}
	\setlength\dashlinedash{1pt}
	\setlength\dashlinegap{2pt}
	\setlength\arrayrulewidth{0.7pt}
	\begin{tabular}{lcc@{\hspace{7mm}}cc@{\hspace{7mm}}ccc}
		\toprule[1pt]
		&&&  \multicolumn{2}{@{}c@{\hspace{7mm}}}{\textbf{Direct}} &\multicolumn{3}{c}{\textbf{Lazy}} \\
		\cmidrule[1pt](l{-1mm}r{6mm}){4-5}
		\cmidrule[1pt](l{-1mm}r{0mm}){6-8}
		\textbf{Instance} & \textbf{Size} & $\boldsymbol{t}_\mathit{abs}$ & $\boldsymbol{t}_\mathit{solve}$ & $\boldsymbol{t}$ &  \textbf{\#Ref} & $\boldsymbol{t}_\mathit{solve}$ & $\boldsymbol{t}$  \\
		\cmidrule[1pt](r{1mm}){1-1}
		\cmidrule[1pt](l{0.5mm}r{1mm}){2-2}
		\cmidrule[1pt](l{0.5mm}r{6mm}){3-3}
		\cmidrule[1pt](l{-1mm}r{1mm}){4-4}
		\cmidrule[1pt](l{0.5mm}r{6mm}){5-5}
		\cmidrule[1pt](l{-1mm}r{1mm}){6-6}
		\cmidrule[1pt](l{0.5mm}r{1mm}){7-7}
		\cmidrule[1pt](l{0.5mm}r{0mm}){8-8}
		NonDetAdd &  4568 & 3.5 & TO & TO & 4  & 1.0  & 4.5  \\
		\hdashline
		CounterSum & 479  & 5.3 & 9.1 & 14.4 & 17  & 0.9 & 6.2  \\
		\hdashline
		AsynchGNI & 437 & 6.1  & 6.9 & 13.0 & 1 & 0.1 &  6.2 \\
		\hdashline
		CompilerOpt1 & 354  & 2.4 &  2.3 & 4.7  & 2  & 0.2  & 2.6 \\
		\hdashline
		CompilerOpt2 & 338  & 2.8 & 2.4 & 5.2 & 2 & 0.2 & 3.0  \\
		\hdashline
		Refine & 1357  & 6.1 & TO & TO & 4 & 0.7 & 6.8  \\
		\hdashline
		Refine2 & 1476  & 5.6 & TO & TO & 5 & 0.6 &  6.2 \\
		\hdashline
		Smaller& 327 & 2.3 & 4.0 & 6.3  & 11  & 0.4  & 2.7  \\
		\hdashline
		CounterDiff  & 959  & 8.5 & 18.3 & 26.8  & 19  & 1.1  & 9.6  \\
		\hdashline
		\Cref{fig:exist}  & 3180  & 11.1 & TO & TO & 22 & 2.9 & 14.0  \\
		\hdashline
		P1 (simple) & 83  & 2.0 & 1.4 & 3.4  & 1 & 0.1 & 2.1  \\
		\hdashline
		P1 (GNI) &  34793  & 17.0 & TO & TO & 72  & 95.7 & 112.7  \\
		\hdashline
		P2 (GNI)  & 15753  & 10.2 & TO & TO & 7 & 5.1 & 15.3  \\
		\hdashline
		P3 (GNI)  & 1429  & 6.6 &  20.9 & 27.5 &  7 & 0.6 & 7.2  \\
		\hdashline
		P4 (GNI)  & 7505  & 16.5 & TO & TO & 72 & 13.2 & 29.7  \\
		\bottomrule[1pt]
	\end{tabular}
\end{table}

\paragraph{Evaluation Beyond $k$-Safety.}

The main novelty of \texttt{HyPA} lies in its ability to, for the first time, verify temporal properties beyond $k$-safety.
As none of the existing tools can verify such properties, we compiled a collection of very small example programs and $\forall^*\exists^*$-safety properties.
Additionally, we modified the boolean programs from \cite{BeutnerF21} (where they checked GNI on boolean programs) by including data from infinite domains.
The properties we checked range from refinement properties for compiler optimizations, over general refinement of nondeterministic programs, to generalized non-interference.
Verification often requires a non-trivial combination of reduction and witness strategy (as the reduction must, e.g.,  compensate for branches of different lengths).
As before, we provide a set of predicates and let \texttt{HyPA} automatically search for a witness strategy with accompanying reduction. 
We list the results in \Cref{tab:beyondeval}.
To highlight the effectiveness of our inner refinement loop, we apply both a direct (explicit) construction of $\gameEF{\STS}{\varphi}{\pred}$ and the lazy (iterative) solver in \Cref{alg:cap}.
Our lazy solver (\Cref{alg:cap}) clearly outperforms an explicit construction and is often the only method to solve the game in reasonable time. 
In particular, we require very few refinement iterations and therefore also few expensive SMT queries.
Unsurprisingly, the problem of verifying properties beyond $k$-safety becomes much more challenging (compared to  $k$-safety verification) as it involves the \emph{synthesis} of a witness function which is already 2\texttt{EXPTIME}-hard for finite-state systems \cite{PnueliR89,CoenenFST19}. 
We emphasize that no other existing tool can verify any of the benchmarks. 

\section{Related Work}
\label{sec:relatedWork}

\paragraph{Asynchronous Hyperproperties.}

Recently, many logics for the formal specification of asynchronous hyperproperties have been developed \cite{BaumeisterCBFS21,GutsfeldMO21,BozzelliPS21,BeutnerF21}.
Our logic OHyperLTL is closely related to stuttering HyperLTL (HyperLTL$_S$) \cite{BozzelliPS21}.
In HyperLTL$_S$ each temporal operator is endowed with a set of temporal formulas $\Gamma$ and steps where the truth values of all formulas in $\Gamma$ remain unchanged are ignored during the operator's evaluation.
As for most mechanisms used to design asynchronous hyperlogics \cite{BaumeisterCBFS21,GutsfeldMO21,BozzelliPS21}, finite-state model checking of HyperLTL$_S$ is undecidable. 
By contrast, in OHyperLTL, we always observe the trace at a fixed location, which is key for ensuring decidable finite-state model checking.

\paragraph{$k$-Safety Verification.}

The literature on $k$-safety verification is rich.
Many approaches verify $k$-safety by using a form of self-composition \cite{BartheDR11,eilers2019modular,churchill2019semantic,FinkbeinerRS15} and often employ reductions to obtain compositions that are easier to verify.
Our game-based interpretation of a reduction (\Cref{sec:ksafety}) is related to Shemer et al.~\cite{ShemerGSV19}, who study $k$-safety verification within a given predicate abstraction using an enumeration-based solver (see \Cref{sec:ksafety} for a discussion). 
Farzan and Vandikas \cite{FarzanV19} present a counterexample-guided refinement loop that simultaneously searches for a reduction and a proof.
Sousa and Dillig \cite{SousaD16} facilitate reductions at the source-code level in program logic.

\paragraph{$\forall^*\exists^*$-Verification.}

Barthe et al.~\cite{BartheCK13} describe an asymmetric product of the system such that only a subset of the behavior of the second system is preserved, thereby allowing the verification of $\forall^*\exists^*$ properties. 
Constructing an asymmetric product and verifying its correctness (i.e., showing that the product preserves all behavior of the first, universally quantified, system) is challenging.  
Unno et al.~\cite{UnnoTK21} present a constraint-based approach to verify functional (opposed to temporal) $\forall\exists$ properties in infinite-state systems using an extension of constraint Horn clauses called pfwCHC.
The underlying verification approach is orthogonal to ours: pfwCHC allows for a clean separation of the actual verification and verification conditions, whereas our approach combines both. 
For example, our method can prove the existence of a witness strategy without ever formulating precise constraints on the strategy (which seems challenging).
Coenen et al.~\cite{CoenenFST19} introduce the game-based reading of existential quantification to verify temporal $\forall^*\exists^*$ properties in a synchronous and finite-state setting.
By contrast, our work constitutes the first verification method for temporal $\forall^*\exists^*$-safety properties in \emph{infinite-state} systems.
The key to our method is a careful integration of reductions which is not possible in a synchronous setting.
For finite-state systems (where the abstraction is precise) and synchronous specifications (where we observe every step), our method subsumes the one in \cite{CoenenFST19}.
Beutner and Finkbeiner \cite{BeutnerF22} use prophecy variables to ensure that the game-based reading of existential quantification is complete in a finite-state setting.
Automatically constructing prophecies for infinite-state systems is interesting future work. 
Pommellet and Touili \cite{PommelletT18} study the verification of HyperLTL in infinite-state systems arising from pushdown systems.
By contrast, we study verification in infinite-state systems that arise from the infinite variables domains used in software.

\paragraph{Game Solving.}

Our game-based interpretations are naturally related to infinite-state game solving \cite{FarzanK18,BaierCFFJS21,WalkerR14,BeyeneCPR14}.
State-of-the-art solvers for infinite-state games unroll the game \cite{FarzanK18}, use necessary subgoals to inductively split a game into subgames \cite{BaierCFFJS21}, encode the game as a constraint system \cite{BeyeneCPR14}, and iteratively refine the controllable predecessor operator \cite{WalkerR14}.
We tried to encode our verification approach directly as an infinite-state linear-arithmetic game.
However, existing solvers (which, notably, work \emph{without} a user-provided set of predicates) could not solve the resulting game \cite{FarzanK18,BaierCFFJS21}.
Our method for encoding the witness strategy using \emph{restrictions} corresponds to hyper-must edges in general abstract games \cite{AlfaroGJ04,AlfaroR07}.
Our inner refinement loop for solving a game with hyper-must edges without explicitly identifying all edges (\Cref{alg:cap}) is thus also applicable in general abstract games.

\section{Conclusion}

In this work, we have presented the first verification method for temporal hyperproperties beyond $k$-safety in infinite-state systems arising in software.
Our method is based on a game-based interpretation of reductions and existential quantification and allows for mutual dependence of both.
Interesting future directions include the integration of our method in a counter-example guided refinement loop that automatically refines the abstraction and ways to lift the current restriction to temporally safe specifications. 
Moreover, it is interesting to study if, and to what extent, the numerous other methods developed for $k$-safety verification of infinite-state systems (apart from reductions) are applicable to the vast landscape of hyperproperties that lies beyond $k$-safety.

\subsubsection{Acknowledgments}

This work was partially supported by the DFG in project 389792660 (Center for Perspicuous Systems, TRR 248).
R.~Beutner carried out this work as a member of the Saarbrücken Graduate School of Computer Science.

\bibliographystyle{splncs04}
\bibliography{references}

\iffullversion

\appendix

\newpage

\section{Proofs for \Cref{sec:ohyperltl}}\label{app:ohyperltl}

\finiteStateR*
\begin{proof}
	Let $X$ be the set of variables in $\calT$ and let $\mathfrak{D}$ be the \emph{finite} domain of the variables (for simplicity we assume that the domain of all variables is the same).
	An explicit state is then an assignment $\alpha : X \to \mathfrak{D}$. 
	Let $\mathfrak{S}$ be the (finite) set of all explicit states. 
	An (explicit) finite-state transition system is a tuple $\calT = (S, S_0,\rho)$ where $S \subseteq \mathfrak{S}$ is a set of explicit states, $S_0 \subseteq S$ is a set of initial states, and $\rho \subseteq S \times S$ is the transition relation.
	The set of traces $\traces{\STS}$ is defined as expected.
	
	Under the assumption on $\calT$, we can view it as a explicit (and computable) finite-state transition system.
	
	For any observation formula $\xi$ let $O(\xi) \subseteq \mathfrak{S}$ be the set of all states in which $\xi$ holds (which is computeable).
	
	Given $\calT = (S, S_0,\rho)$ and $\xi$ we construct an explicit finite-state transition system $\STS_\xi$ such that $\traces{\STS_\xi} = \{ \filter{t}{\xi} \mid t \in \traces{\STS} \land \filterdef{t}{\xi}\}$.
	We define $\STS_\xi \coloneqq (O(\xi) \cap S, S'_0,\rho')$ where
	\begin{align*}
		S_0' \coloneqq  \big\{s \in O(\xi) \cap S \mid \exists &s_0,s_1,\ldots, s_n \in S\ldot s_0 \in S_0 \land  s_n = s \, \land\\
		&\forall 0 \leq i \leq n - 1\ldot (s_i, s_{i+1}) \in \rho \, \land \\
		&\forall 0 \leq i \leq n - 1\ldot s_i \not\in O(\xi)\big\}.
	\end{align*}
	That is all states where $\xi$ holds that are reachable from some state in $S_0$ in $\calT$ without passing through another state where $\xi$ holds.
	Similarly we define
	\begin{align*}
		\rho' \coloneqq  \big\{(s, s') \in (O(\xi) \cap S )^2 \mid \exists &s_0,s_1,\ldots, s_n \in S\ldot n \geq 1 \land  s_0 = s \land s_n = s' \, \land\\
		&\forall 0 \leq i \leq n - 1\ldot (s_i, s_{i+1}) \in \rho \; \land \\
		&\forall 1 \leq i \leq n - 1\ldot s_i \not\in O(\xi)\big\}.
	\end{align*}
	I.e., there is an edge $(s, s')$ only if $s$ and $s'$ are connected by a path of length at least $1$ of unobserved states in $S$.

	As $\traces{\STS_\xi} = \{ \filter{t}{\xi} \mid t \in \traces{\STS} \land \filterdef{t}{\xi}\}$ we can reduce OHyperLTL model checking on $\STS$ to HyperLTL model checking on $\STS_\xi$, which is decidable \cite{FinkbeinerRS15}.
	Note that a OHyperLTL can use different observation formulas for different quantifiers. 
	In the resulting HyperLTL model checking instance we thus need to resolve different quantifiers on different systems (as, in general, $\STS_\xi \neq \STS_{\xi'}$ when $\xi \not\equiv \xi'$), which is easily done.
	\qed
\end{proof}

\newpage
\section{Proofs for \Cref{sec:ksafety}}\label{app:ksafety}

\soundnessksafety*
\begin{proof}
	Let $\STS = (X, \mathit{init}, \mathit{step})$ be the STS, let $\varphi = \forall \pi_1 : \xi_1. \ldots \forall \pi_k : \xi_k\ldot \phi$ be the OHyperLTL formula, and let $\calA_\phi$ be the deterministic safety automaton for $\phi$ used in the construction of $\gameE{\STS}{\varphi}{\pred}$.
	Assume that $\sigma$ is a winning strategy for $\safe$ in $\gameE{\STS}{\varphi}{\pred}$.
	We show that $\STS \models \varphi$.
	For this, let $t_1, \ldots, t_k \in \traces{\STS}$ be arbitrary traces such that $\filterdef{t_i}{\xi_i}$ for every $i$.
	We show that $[\pi_1 \mapsto \filter{t_1}{\xi_1}, \ldots, \pi_k \mapsto \filter{t_k}{\xi_k}] \models \phi$.
	
	The idea is to (implicitly) stutter traces $t_1, \ldots, t_k$ between two observation points and compute a pointwise abstraction for these stuttered traces.
	The stuttering is dictated by $\sigma$, i.e., we simulate prefixes in the game and query $\sigma$ to determine which scheduling to pick. All non-scheduled copies are stuttered.
	If $\sigma$ picks a scheduling $M$, the refuter can pick (in $\gameE{\STS}{\varphi}{\pred}$) an abstract successor state that is compatible with $M$.
	In our simulation we pick the exact abstract states that arises when moving as defined by traces $t_1, \ldots, t_k$, which, by definition of $M$-successor, is a valid step in the game.
	This simulation thus gives an abstraction of the stuttered traces which (as $\sigma$ is winning) avoids a visit to losing states in $\calA_\phi$.
	In the following we give a more detailed description of this high-level strategy.
	
	For $k$ assignment to $X$, $\mu_1, \ldots, \mu_k$ we write $\mu_1 \otimes \cdots \otimes \mu_k$ for the assignment over $\vec{X}$ defined by $(\mu_1 \otimes \cdots \otimes \mu_k)(x_{\pi_i}) \coloneqq \mu_i (x)$.
	For each assignment $\mu$ to $\vec{X}$ we define $\mathit{Abstract}(\mu)$ as the unique abstract state $\hat{s}$ such that $\mu \models \predSat{\hat{s}}$.
	We zip the traces $\filter{t_1}{\xi_1}, \ldots, \filter{t_k}{\xi_k}$ into a single trace $\overline{t}$ over assignments to $\vec{X}$ by defining $\overline{t}(j) \coloneqq \filter{t_1}{\xi_1}(j) \otimes \cdots \otimes \filter{t_k}{\xi_k}(j)$.
	To show that $[\pi_1 \mapsto \filter{t_1}{\xi_1}, \ldots, \pi_k \mapsto \filter{t_k}{\xi_k}] \models \phi$ it suffices to show that the unique run of $\calA_\phi$ on $\overline{t}$ does not visit a bad state.
	
	Consider the construction in \Cref{fig:construction1}.
	Note that this construction will never finish but allows us to point to key steps showing that  $[\pi_1 \mapsto \filter{t_1}{\xi_1}, \ldots, \pi_k \mapsto \filter{t_k}{\xi_k}] \models \phi$.
	 
	\floatstyle{plainruled}
	\restylefloat{algorithm}

	\begin{figure}[t]
		\begin{algorithm}[H]
			\begin{algorithmic}[1]
				\State $\mu_i \leftarrow t_i(0)$ for $1 \leq i \leq k$
				\State $c_i \leftarrow 0$ for $1 \leq i \leq k$
				\State $\hat{s} \leftarrow \mathit{Abstract}(\mu_1 \otimes \cdots \otimes \mu_k)$
				\State $q \leftarrow q_{\phi, 0}$
				\State $b \leftarrow \top^k$
				\While{true}
				\If{$b = \top^k \land \mathit{obs}(\hat{s}) = \top^k$ }
					\State $\mu'_i \leftarrow \mu_i$ for $1 \leq i \leq k$
					\State $c'_i \leftarrow c_i$ for $1 \leq i \leq k$
					\State $\hat{s}' \leftarrow \hat{s}$
					\State $q' \leftarrow \delta_\phi(q, \hat{s})$
					\State $b' \leftarrow \bot^k$
				\Else
					\State $(\_,\_,\_,M) \leftarrow \sigma(\hat{s}, q, b)$
					\State $c'_i \leftarrow \mathit{ite}(\pi_i \in M, c_i + 1, c_i)$ for $1 \leq i \leq k$
					\State $\mu'_i \leftarrow \mathit{ite}(\pi_i \in M, t_i(c_i), \mu_i)$ for $1 \leq i \leq k$
					\State $\hat{s}' \leftarrow \mathit{Abstract}(\mu'_1 \otimes \cdots \otimes \mu'_k)$
					\State $q' \leftarrow q$
					\State $b' \leftarrow b[i \mapsto \top]_{\pi_i \in M}$
				\EndIf
				\State $\mu_i \leftarrow \mu_i'$ for $1 \leq i \leq k$
				\State $c_i \leftarrow c_i'$ for $1 \leq i \leq k$
				\State $\hat{s}\leftarrow \hat{s}'$
				\State $q \leftarrow q'$
				\State $b \leftarrow b'$
				\EndWhile
			\end{algorithmic}
		\end{algorithm}
		\vspace{-7mm}
		\caption[]{Construction for the proof of \Cref{theo:soundessKSafety}.\label{fig:construction1}}
	\end{figure}
	
	We maintain a concrete state $\mu_i$ for each copy $1 \leq i \leq k$, initially set to the initial state according to the fixed traces (line 1).
	Additionally we maintain a counter $c_i$ that tracks at which position of $t_i$ the current state is located.
	It will always be the case that $\mu_i = t_i(c_i)$.
	
	The construction then simulates a play in $\gameE{\STS}{\varphi}{\pred}$ using the winning strategy $\sigma$ to resolve choices made by player $\safe$ as follows:
	The simulation starts in state $(\hat{s}, q, b)$ where $\hat{s}$ is the initial abstract state based on $t_1, \ldots, t_k$, $q$ the initial state of $\calA_\phi$, and $b = \top^k$ (lines 3 - 5).
	Note that, by construction, this is an initial game state in $\gameE{\STS}{\varphi}{\pred}$.
	
	The simulation then continuously advances the play. 
	If $b = \top^k$ and $\mathit{obs}(\hat{s}) = \top^k$ it applies transition rule \textbf{(2)} in lines 8-12. 
	In particular, the concrete states $\mu_i$, counters $c_i$ and abstract state $\hat{s}$ remain unchanged.
	Otherwise it queries $\sigma$ on the current state $(\hat{s}, q, b)$.
	By transition rule \textbf{(1)}, $\sigma$ can only select a scheduling $M$ (and not change the other state components).
	In line 14 we write $\_$ to mark that we do not care about a value.
	Each trace $t_i$ where $\pi_i \in M$ then takes a step, i.e., we increment $c_i$ for those copies, and we update the current state $\mu_i$ (lines 15 and 16).
	For non-scheduled copies $c_i$ is left unchanged and so is the concrete state $\mu_i$.
	We compute the new abstract state $\hat{s}'$ and update $b$ (as in transition rule \textbf{(3)}).
	
	We can establish a few basic properties:
	
	\begin{itemize}
	\item 	
		\textbf{P1:} \textit{Whenever the loop is entered, $\mu_1 \otimes \cdots \otimes \mu_k \models \predSat{\hat{s}}$.}
		
		This follows directly from the construction and the definition of $\mathit{Abstract}(\cdot)$.
		
	\item 
		\textbf{P2:} \textit{Whenever the loop is entered, $\mu_i = t_i(c_i)$ for all $1 \leq i \leq k$.}

		This follows directly from the construction and the updates performed in lines 15 and 16.
	
	\item 
		\textbf{P3:} \textit{Let $(\hat{s}, q, b)$ be the current state at the beginning of a loop body and $(\hat{s}', q', b')$ after the loop body has executed once.
					Then, in $\gameE{\STS}{\varphi}{\pred}$, player $\reach$ can force a play from $(\hat{s}, q, b)$ to  $(\hat{s}', q', b')$ when the safety player follows strategy $\sigma$.}
		
		In case the conditional in line 7 is taken this is trivial as it directly corresponds to transition rule \textbf{(2)} of the game (the only one that is applicable in that case).
		The more interesting direction is thus the case where lines 14-19 are executed. 
		So let $(\hat{s}, q, b)$ be a game state and let $\mu_1, \ldots, \mu_k$ be the current concrete states at the beginning of a loop iteration.
		By \textbf{P1} we get that $\mu_1 \otimes \cdots \otimes \mu_k \models \predSat{\hat{s}}$.
		As $(\hat{s}, q, b, M) = \sigma(\hat{s}, q, b)$, game state $(\hat{s}, q, b, M)$ is reachable under $\sigma$ (using transition rule \textbf{(1)}).
		By \textbf{P2} we have that $\mu_i = t_i(c_i)$ for each $i$.
		As each $t_i$ is a trace in $\STS$ we get that $t_i(c_i) \uplus (t_i(c_i + 1))' \models \mathit{step}$.
		Now $\mu_i' = \mu_i$ for all non-scheduled copies $\pi_i \not\in M$, and $\mu_i \uplus \mu_i' \models \mathit{step}$ for all $\pi_i \in M$.
		It is therefore easy to see that $\hat{s} \xrightarrow{M} \hat{s}'$ (follows directly from the definition of $\xrightarrow{M}$).
		The game state $(\hat{s}', q', b')$ reached after the loop is thus a successor state of $(\hat{s}, q, b, M)$ via transition rule \textbf{(3)}.
		So $(\hat{s}, q, b)$ can take a step to $(\hat{s}, q, b, M)$ that is fixed by $\sigma$ and from $(\hat{s}, q, b, M)$ player $\reach$ can move the game to $(\hat{s}', q', b')$ as required.
	\end{itemize}

	To show that $[\pi_1 \mapsto \filter{t_1}{\xi_1}, \ldots, \pi_k \mapsto \filter{t_k}{\xi_k}] \models \phi$ (or equivalently,  that the unique run of $\calA_\phi$ on $\overline{t}$ does not visit a bad state) consider the following:
	We focus on those iterations where the conditional on line 7 is taken. 
	We mark these with a superscript.
	That is $\mu_i^j$ is the value of $\mu_i$ when conditional in line 7 was taken for the $j$th time, and similarly $c_i^j$, $\hat{s}^j$, and $q^j$ (for the values of $c_i$, $\hat{s}$ and $q$, respectively).
	The important observation is that $c_i^j$ is now exactly the $j$th index where $t_i$ satisfies the observation formula $\xi_i$, i.e., $\filter{t_i}{\xi_i}(j) = t_i(c_i^j)$.
	This holds by the design of $\gameE{\STS}{\varphi}{\pred}$, i.e., player $\safe$ can only schedule copies that have not moved yet or have reached an observation point (via transition rule \textbf{(1)}).
	We thus get 
	\begin{align*}
		\overline{t}(j) = \filter{t_1}{\xi_1}(j) \otimes \cdots \otimes \filter{t_k}{\xi_k}(j) = t_1(c_1^j) \otimes \cdots \otimes t_k(c_k^j)
	\end{align*}
	By \textbf{P2},  $t_i(c_i^j) = \mu_i^j$ and by \textbf{P1} $\mu_1^j \otimes \cdots \otimes \mu_k^j \models \predSat{\hat{s}^j}$.
	So $\overline{t}(j) \models \hat{s}^j$ for every $j \in \nat$, i.e., the sequence $\hat{s}^0,\hat{s}^1,\ldots$ forms a pointwise abstraction of $\overline{t}$.
	Moreover $q^0,q^1,\ldots$ is the unique run of $\calA_\phi$ on $\hat{s}^0,\hat{s}^1,\ldots$.
	As $\sigma$ is winning and the construction simulates an actual game play allowed by $\sigma$  (as stated in \textbf{P3}), we get that all states $q^0,q^1,\ldots$ are not losing. 
	This already concludes that $\overline{t}$ is accepted by $\calA_\phi$ (does not have a rejecting run to a losing state), and so $[\pi_1 \mapsto \filter{t_1}{\xi_1}, \ldots, \pi_k \mapsto \filter{t_k}{\xi_k}] \models \phi$ as required.
	\qed
\end{proof}

\newpage

\section{Proofs for \Cref{sec:beyondksafety}}\label{app:beyond}

\soundessBeyond*
\begin{proof}
	Let $\STS = (X, \mathit{init}, \mathit{step})$ be the STS, let $\varphi = \forall \pi_1 : \xi_1. \ldots \forall \pi_l : \xi_l. \exists \pi_{l+1} : \xi_{l+1}. \cdots \exists \pi_k : \xi_k \ldot \phi$ be the OHyperLTL formula, and let $\calA_\phi$ be the deterministic safety automaton for $\phi$ used in the construction of $\gameEF{\STS}{\varphi}{\pred}$.
	Assume that $\sigma$ is a winning strategy for $\safe$ in $\gameEF{\STS}{\varphi}{\pred}$.
	We show that $\STS \models \varphi$.
	For this, let $t_1, \ldots, t_l \in \traces{\STS}$ be arbitrary traces such that $\filterdef{t_i}{\xi_i}$ for every $1 \leq i \leq l$.
	We will construct traces $t_{l+1}, \ldots, t_k \in \traces{\STS}$ such that $[\pi_1 \mapsto \filter{t_1}{\xi_1}, \ldots, \pi_k \mapsto \filter{t_k}{\xi_k}] \models \phi$. 

	To do so we follow the idea used in the proof of \Cref*{theo:soundessKSafety} and query $\sigma$ to select a scheduling. 
	In addition, we need to actually construct traces $t_{l+1}, \ldots, t_k$. The idea is to use the definition of $\validRes{\hat{s}}{M}{A}$:
	In each step, we plug the current concrete states in the universal quantifiers of $\validRes{\hat{s}}{M}{A}$ and obtain a concrete witness for the existentially quantified variables.
	These correspond exactly to the successor states for the traces $t_{l+1}, \ldots, t_k$.
	In the following we give a more detailed construction.

	For simplicity we assume that there is a unique concrete state $\mu_{\mathit{init}}$ that satisfies $\mu_{\mathit{init}} \models \mathit{init}$.
	Note that this implies that $t_i(0) = \mu_{\mathit{init}}$ for all $1 \leq i \leq l$.
	Let $\hat{s}_{\mathit{init}}$ be the resulting initial abstract state if all $k$ copies are in $\mu_{\mathit{init}}$.
	Consider the construction in \Cref{fig:construction2} (which, again, does never terminate but allows the construction of witness traces in the limit).

	\floatstyle{plainruled}
	\restylefloat{algorithm}

	\begin{figure}[t]
		\begin{algorithm}[H]
			\begin{algorithmic}[1]
				\State $\mu_i \leftarrow \mu_{\mathit{init}}$ for $1 \leq i \leq k$
				\State $c_i \leftarrow 0$ for $1 \leq i \leq l$
				\State $\hat{s} \leftarrow \hat{s}_{\mathit{init}}$
				\State $q \leftarrow q_{\phi, 0}$
				\State $b \leftarrow \top^k$
				\State $t_i = [\mu_{\mathit{init}}]$ for $l+1 \leq i \leq k$
				\While{true}
				\If{$b = \top^k \land \mathit{obs}(\hat{s}) = \top^k$ }
					\State $\mu'_i \leftarrow \mu_i$ for $1 \leq i \leq k$
					\State $c'_i \leftarrow c_i$ for $1 \leq i \leq k$
					\State $\hat{s}' \leftarrow \hat{s}$
					\State $q' \leftarrow \delta_\phi(q, \hat{s})$
					\State $b' \leftarrow \bot^k$
				\Else
					\State $(\_,\_,\_,M, A) \leftarrow \sigma(\hat{s}, q, b)$
					\State $c'_i \leftarrow \mathit{ite}(\pi_i \in M, c_i + 1, c_i)$ for $1 \leq i \leq l$
					\State $\mu'_i \leftarrow \mathit{ite}(\pi_i \in M, t_i(c_i), \mu_i)$ for $1 \leq i \leq l$
					\State $\mu'_{l+1}, \ldots, \mu'_k \leftarrow \mathit{extractModel}(\validRes{\hat{s}}{M}{A}, \{\mu_i\}_{i=1}^k, \{\mu'_i\}_{i=1}^l)$ 
					
					\State $\hat{s}' \leftarrow \mathit{Abstract}(\mu'_1 \otimes \cdots \otimes \mu'_k)$
					\State $q' \leftarrow q$
					\State $b' \leftarrow b[i \mapsto \top]_{\pi_i \in M}$
					\State $t_i = t_i + [\mu'_i]$ for $l+1 \leq i \leq k$ with $\pi_i \in M$
				\EndIf
				\State $\mu_i \leftarrow \mu_i'$ for $1 \leq i \leq k$
				\State $c_i \leftarrow c_i'$ for $1 \leq i \leq l$
				\State $\hat{s}\leftarrow \hat{s}'$
				\State $q \leftarrow q'$
				\State $b \leftarrow b'$
				\EndWhile
			\end{algorithmic}
		\end{algorithm}
		
		\vspace{-7mm}
		\caption[]{Construction for the proof of \Cref{theo:soundnessBeyond}.\label{fig:construction2}}
	\end{figure}

	The basic construction is similar to that in the proof of \Cref{theo:soundessKSafety}.
	We again maintain $\hat{s}, q$, and $b$ and simulate the game using $\sigma$ to resolve choices made by player $\safe$.
	Additionally, we maintain a trace $t_i$ for each $l+1 \leq i \leq k$, initially set to the length-$1$-trace consisting of $\mu_{\mathit{init}}$.
	
	In lines 9-13 we then perform transition step \textbf{(2)} whenever this is possible.
	If not, we query $\sigma$ to determine a scheduling $M$ and restriction $A$ (in line 15).
	We update the concrete state of universally quantified executions that are scheduled (line 16 and 17).
	So far, this is identical to the construction in the proof of \Cref{theo:soundessKSafety}.
	The crucial point is that we need to fix a next concrete state for the existentially quantified traces that are scheduled.
	Here we make use of the definition of $\validRes{\hat{s}}{M}{A}$. 
	Recall that $\validRes{\hat{s}}{M}{A}$ is defined as:
	\begin{align*}
        &\forall \{X_{\pi_i}\}_{i=1}^{k}. \forall \{X_{\pi_i}'\}_{i=1}^l. \; \predSat{\hat{s}} \land \bigwedge_{i=1}^l \mathit{ite}\Big(\pi_i \in M, \mathit{step}_{\langle \pi_i \rangle}, \bigwedge_{x \in X} x_{\pi_i}' = x_{\pi_i}\Big)\\[-0.3cm]
        &\quad\Rightarrow \exists \{X_{\pi_i}'\}_{i=l+1}^k. \bigwedge_{i=l+1}^k \mathit{ite}\Big(\pi_i \in M, \mathit{step}_{\langle \pi_i \rangle}, \bigwedge_{x \in X} x_{\pi_i}' = x_{\pi_i}\Big) \land \bigvee\limits_{\hat{s}' \in A} \predSat{\hat{s}'}^{\langle'\rangle}
	\end{align*}
	
	We now know that $\mu_1 \otimes \cdots \otimes \mu_k \models \predSat{\hat{s}}$ at all times (similar to \textbf{P1} in the proof of \Cref{theo:soundessKSafety}).
	Moreover when using assignments $\mu_1, \ldots, \mu_k$ for $\{X_{\pi_i}\}_{i=1}^{k}$ and $\mu'_1, \ldots, \mu'_l$ for $\{X_{\pi_i}'\}_{i=1}^l$ we get a satisfying model for 
	\begin{align*}
		\bigwedge_{i=1}^l \mathit{ite}\Big(\pi_i \in M, \mathit{step}_{\langle \pi_i \rangle}, \bigwedge_{x \in X} x_{\pi_i}' = x_{\pi_i}\Big),
	\end{align*}
	i.e., the premise in $\validRes{\hat{s}}{M}{A}$.
	This holds as $t_1, \ldots, t_l$ are actual traces and lines 16 and 17 only move those universally quantified copies that are actually scheduled (see the proof of \Cref{theo:soundessKSafety}).

	As $\validRes{\hat{s}}{M}{A}$ holds (by transition rule \textbf{(1)}) we thus get a concrete assignments $\mu'_{l+1}, \ldots, \mu'_k$ to $\{X_{\pi_i}'\}_{i=l+1}^k$ that satisfies 
	\begin{align*}
		\bigwedge_{i=l+1}^k \mathit{ite}\Big(\pi_i \in M, \mathit{step}_{\langle \pi_i \rangle}, \bigwedge_{x \in X} x_{\pi_i}' = x_{\pi_i}\Big) \land \bigvee\limits_{\hat{s}' \in A} \predSat{\hat{s}'}^{\langle'\rangle}
	\end{align*}
	i.e., the conclusion of $\validRes{\hat{s}}{M}{A}$ (together with $\mu'_1, \ldots, \mu'_l$).
	We extract these assignments in line 18.
	That is, we plug in the concrete assignments $\mu_1, \ldots, \mu_k$ and $\mu'_1, \ldots, \mu'_l$ and get assignments $\mu'_{l+1}, \ldots, \mu'_k$.
	For each copy $l+1 \leq i \leq k$ that was scheduled in $M$ we then update the existentially quantified traces with this newly obtained assignment (line 22).
	By definition of $\validRes{\hat{s}}{M}{A}$ this will construct traces $t_{l+1}, \ldots, t_k$ that are contained in $\traces{\STS}$.

	We claim that, in the limit, the traces $t_{l+1},\ldots, t_k$ constructed serve as witness traces for $t_1, \ldots, t_l$.
	The crucial point is that the construction essentially simulates a play in $\gameEF{\STS}{\varphi}{\pred}$ that is compatible with $\sigma$.
	In particular, note that, by construction of the assignments $\mu'_{l+1},\ldots, \mu'_k$ from $\validRes{\hat{s}}{M}{A}$ we get that $\hat{s}'$ (as computed in line 19) satisfies $\hat{s}' \in A$, i.e., the abstract state is contained in the restriction chosen by $A$.
	The game state $(\hat{s}', q', b')$ is thus a successor state of $(\hat{s}, q, b, M, A)$ by transition rule \textbf{(3)}.

	The proof that $[\pi_1 \mapsto \filter{t_1}{\xi_1}, \ldots, \pi_k \mapsto \filter{t_k}{\xi_k}] \models \phi$ is then analogous to the proof of \Cref{theo:soundessKSafety}.
	\qed
\end{proof}

\fi

\end{document}